\documentclass{elsarticle_arxiv}

\usepackage{lineno,hyperref}
\modulolinenumbers[5]

\usepackage{multirow}

\usepackage{tcolorbox}
\usepackage{amsmath}
\usepackage{amsfonts}
\usepackage{verbatim}
\usepackage{bbm}

\usepackage{caption}

\newtheorem{theorem}{Theorem}

\newtheorem{lemma}{Lemma}
\newtheorem{definition}{Definition}
\newtheorem{proof}{Proof}
\usepackage{todonotes}
\usepackage{cancel}

\usepackage{xspace}

\usepackage{array} 
\newcolumntype{P}[1]{>{\centering\arraybackslash}p{#1}} 

\def\dosth#1{\ifx###1##\else\dofirst#1\anytoken\fi}
\def\doagain#1\anytoken{\dosth{#1}}
\def\payoffpairs#1#2#3{\m=#1\multiply\m by 4 \advance\m by -1 \n=1
  \def\dofirst##1{\put(\n,-\m){\makebox(0,0){\strut##1}}\advance\n by 4 \doagain}%
  \dosth{#2\strut}%
  \m=#1\multiply\m by 4 \advance\m by -3 \n=3 \dosth{#3\strut}}
\def\singlepayoffs#1#2{\m=#1\multiply\m by 4 \advance\m by -2 \n=2
  \def\dofirst##1{\put(\n,-\m){\makebox(0,0){\strut##1}}\advance\n by 4 \doagain}%
  {\large\dosth{#2\strut}}}
\newcommand{\bimatrixgame}[8]{%
\setlength{\unitlength}{#1}%
\newcount\rows
\newcount\cols
\rows=#2
\cols=#3
\newcount\rowcoord
\newcount\colcoord
\rowcoord=\rows
\colcoord=\cols
\multiply\rowcoord by 4
\multiply\colcoord by 4
\newcount\m
\newcount\n
\m=\rowcoord
\n=\colcoord
\advance\m by 2 
\advance\n by 2 
\begin{picture}(\n,\m)(-2,-\rowcoord)
\m=\rows
\n=\cols
\advance\m by 1
\advance\n by 1 
\thinlines
\multiput(0,0)(0,-4){\m}{\line(1,0){\colcoord}}
\multiput(0,0)(4,0){\n}{\line(0,-1){\rowcoord}}
\put(0,0){\line(-1,1){2}}
\put(-1.5,0.5){\makebox(0,0)[r]{#4}}  
\put(-.7,1.7){\makebox(0,0)[l]{#5}}   
\n=2
\def\dofirst##1{\put(-0.8,-\n){\makebox(0,0)[r]{\strut##1}}\advance\n by 4
   \doagain}
\dosth{#6\strut} 
\n=2
\def\dofirst##1{\put(\n,1.0){\makebox(0,0){\strut##1}}\advance\n by 4
   \doagain}
\dosth{#7\strut}#8%
\end{picture}}
\newcounter{probc}
\newcommand{\probc}[1]{\refstepcounter{probc}\label{#1}}

\newcommand{\PPAD}{\ensuremath{\mathtt{PPAD}}\xspace}
\newcommand{\NP}{\ensuremath{\mathtt{NP}}\xspace}

\newcommand{\ETR}{\ensuremath{\mathtt{ETR}}\xspace}
\newcommand{\EOL}{\textsc{EndOfTheLine}\xspace}

\newcommand{\eps}{\ensuremath{\epsilon}\xspace}

\newcommand{\supp}{\mathrm{supp}}

\newcommand{\freegame}{\ensuremath{\textsc{FreeGame}_\delta}\xspace}
\newcommand{\mone}{\ensuremath{\mathrm{Merlin}_1}\xspace}
\newcommand{\mtwo}{\ensuremath{\mathrm{Merlin}_2}\xspace}
\newcommand{\meye}{\ensuremath{\mathrm{Merlin}_i}\xspace}
\newcommand{\art}{\ensuremath{\mathrm{Arthur}}\xspace}

\newcommand{\xbf}{\ensuremath{\mathbf{x}}\xspace}
\newcommand{\ybf}{\ensuremath{\mathbf{y}}\xspace}
\newcommand{\ubf}{\ensuremath{\mathbf{u}}\xspace}
\newcommand{\dbf}{\ensuremath{\mathbf{d}}\xspace}

\newcommand{\gcal}{\ensuremath{\mathcal{G}}\xspace}
\newcommand{\dcal}{\ensuremath{\mathcal{D}}\xspace}

\newcommand{\tcal}{\ensuremath{\mathcal{T}}\xspace}
\newcommand{\fcal}{\ensuremath{\mathcal{F}}\xspace}
\newcommand{\ucal}{\ensuremath{\mathcal{U}}\xspace}

\newcommand{\gcalp}{\ensuremath{\mathcal{G}'}\xspace}
\newcommand{\gcalpp}{\ensuremath{\mathcal{G}''}\xspace}

\newcommand{\ifrak}{\ensuremath{\mathfrak{i}}\xspace}
\newcommand{\jfrak}{\ensuremath{\mathfrak{j}}\xspace}
\newcommand{\ifrakp}{\ensuremath{\mathfrak{i}'}\xspace}
\newcommand{\jfrakp}{\ensuremath{\mathfrak{j}'}\xspace}

\newcommand{\ufrak}{\ensuremath{\mathfrak{u}}\xspace}

\newcommand{\xbfi}{\ensuremath{\xbf_\ifrak}\xspace}
\newcommand{\ybfj}{\ensuremath{\ybf_\jfrak}\xspace}

\newcommand{\eeq}{\ensuremath{E^{\eps}}\xspace}

\newcommand{\txy}{\ensuremath{\mathcal{T}_{(\xbf, \ybf)}}\xspace}
\DeclareMathOperator{\sw}{SW\xspace}
\DeclareMathOperator{\bsw}{BSW\xspace}
\DeclareMathOperator{\sat}{SAT\xspace}
\DeclareMathOperator{\polylog}{polylog}

\begin{document}

\begin{frontmatter}

\title{Inapproximability Results for Approximate Nash Equilibria}
\tnotetext[mytitlenote]{The authors were supported by EPSRC grant EP/L011018/1. 
A short version of this paper appeared at WINE 2016.}


\author[affil1]{Argyrios Deligkas}
\author[affil2]{John Fearnley}
\author[affil2]{Rahul Savani}

%
\address[affil1]{Technion-Israel Institute of Technology}
\address[affil2]{University of Liverpool}
\begin{abstract}
We study the problem of finding approximate Nash equilibria that satisfy certain
conditions, such as providing good social welfare. In particular, we study the
problem $\eps$-NE $\delta$-SW: find an $\eps$-approximate
Nash equilibrium ($\eps$-NE) that is within $\delta$ of the best social welfare
achievable by an $\eps$-NE. Our main result is that, if the
exponential-time hypothesis (ETH) is true, then solving
$\left(\frac{1}{8} - \mathrm{O}(\delta)\right)$-NE $\mathrm{O}(\delta)$-SW for an $n\times n$
bimatrix game requires $n^{\mathrm{\widetilde \Omega}(\log n)}$ time.
Building on this result, we show similar conditional running time lower bounds
on a number of decision problems for approximate Nash equilibria that do not
involve social welfare, including maximizing or minimizing a certain player's
payoff, or finding approximate equilibria contained in a given pair of supports.
We show quasi-polynomial lower bounds for these problems assuming that ETH
holds, where these lower bounds
apply to $\eps$-Nash
equilibria for all $\eps < \frac{1}{8}$. The hardness
of these other decision problems has so far only been studied in the context of
exact equilibria.

%
%
\end{abstract}


\end{frontmatter}

\section{Introduction}
\label{sec:intro}

One of the most fundamental problems in game theory is to find a Nash
equilibrium of a game. Often, we are not interested in finding any Nash
equilibrium, but instead we want to find one that also satisfies
certain constraints. For example, we may want to find a Nash equilibrium that
provides high \emph{social welfare}, which is the sum of the players' payoffs.

In this paper we study such problems for \emph{bimatrix games}, which are
two-player strategic-form games. Unfortunately, for bimatrix games, it is known
that these problems are hard. Finding any Nash equilibrium of a bimatrix
game is \PPAD-complete~\cite{DGP09}, while finding a constrained Nash
equilibrium turns out to be even harder. Gilboa and Zemel~\cite{GZ89} studied
several decision problems related to Nash equilibria.  They proved that it is
\NP-complete to decide whether there exist Nash equilibria in  bimatrix games
with some ``desirable'' properties, such as high social welfare. Conitzer and
Sandholm~\cite{CS08} extended the list of \NP-complete problems of~\cite{GZ89}
and furthermore proved inapproximability results for some of them. Recently,
Garg et al.~\cite{GMVY} and Bilo and Mavronicolas~\cite{BM16,BM17} extended these
results to many player games and provided \ETR-completeness results for them.

\noindent \textbf{Approximate equilibria.}
Due to the apparent hardness of finding exact Nash equilibria, focus has 
shifted to \emph{approximate} equilibria. There are two natural notions of
approximate equilibrium, both of which will be studied in this paper. An
\emph{$\epsilon$-approximate Nash equilibrium} ($\epsilon$-NE) requires that
each player has an expected payoff that is within $\epsilon$ of their best
response payoff. An \emph{$\epsilon$-well-supported Nash equilibrium}
($\epsilon$-WSNE) requires that both players only play strategies whose payoff
is within $\epsilon$ of the best response payoff. Every $\epsilon$-WSNE is an
$\epsilon$-NE but the converse does not hold, so a WSNE is a more restrictive
notion.

There has been a long line of work on finding approximate equilibria
\cite{BBM10,CDFFJS,DMP07,DMP09,FGSS12,KS10,TS08}.
Since we use an additive notion of approximation, it is
common to rescale the game so that the payoffs lie in $[0, 1]$, which allows
different algorithms to be compared. 
The state of the art for polynomial-time algorithms is the following.
There is a polynomial-time algorithm that computes an $0.3393$-NE~\cite{TS08},
and a polynomial-time algorithm that computes a $0.6528$-WSNE~\cite{CDFFJS}.
 
There is also a \emph{quasi-polynomial time approximation
scheme} (QPTAS) for finding approximate Nash equilibria. The algorithm of
Lipton, Markakis, and Mehta finds an $\epsilon$-NE in $n^{O(\frac{\log
n}{\epsilon^2})}$ time~\cite{LMM03}. 
They proved that there is always an
$\epsilon$-NE with support of logarithmic size, and then they use a brute-force search over
all possible candidates to find one.
We will refer to their algorithm as the LMM algorithm. 
 
A recent breakthrough of Rubinstein implies that we cannot do better than a
QPTAS like the LMM algorithm~\cite{R16}: assuming an exponential time hypothesis for \PPAD (PETH),
there is a small constant, $\eps^*$, such that for $\eps < \eps^*$, every
algorithm for finding an \eps-NE requires quasi-polynomial time. Briefly, PETH
is the conjecture that \EOL, the canonical \PPAD-complete problem,
cannot be solved faster than exponential time.


\noindent \textbf{Constrained approximate Nash equilibria.}
While deciding whether a game has an exact Nash equilibrium that satisfies
certain constraints is \NP-hard for most interesting constraints, this is not
the case for approximate equilibria, because the LMM algorithm can be
adapted to provide a QPTAS for them. 
The question then arises whether one can do better.
 
Let the problem $\eps$-NE $\delta$-SW be the problem
of finding an $\epsilon$-NE whose social welfare is within $\delta$ of the best
social welfare that can be achieved by an $\eps$-NE.
Hazan and Krauthgamer~\cite{HK11} and Austrin,
Braverman and Chlamtac~\cite{ABC13} proved that there is a small but constant
$\eps$ such that $\eps$-NE $\eps$-SW is at least as hard as finding a hidden
clique of size $\mathrm{O}(\log n)$ in the random graph $G_{n, 1/2}$. 
This was further strengthened by Braverman, 
Ko, and Weinstein~\cite{BKW15} who showed a lower bound based on the
\emph{exponential-time hypothesis} (ETH), which is the conjecture that any
deterministic algorithm for 3SAT
requires $2^{\mathrm{\Omega}(n)}$ time. More precisely, they showed
%
%
that under ETH there is a small constant \eps
such that any algorithm for $\mathrm{O}(\eps)$-NE $\mathrm{O}(\eps)$-SW%
\footnote{While the proof in~\cite{BKW15} produces a lower bound for $0.8$-NE
$(1 - \mathrm{O}(\eps))$-SW, this is in a game with maximum payoff $\mathrm{O}(1/\eps)$.
Therefore, when the
payoffs in this game are rescaled to $[0, 1]$, the resulting lower bound only
applies to $\eps$-NE $\eps$-SW.}
requires $n^{\text{poly}(\eps) \log(n)^{1 - o(1)}}$ time%
\footnote{Although the paper claims
that they obtain a $n^{\mathrm{\widetilde O}(\log n)}$ lower bound, the proof reduces
from the \emph{low error} result from~\cite{AIM14} (cf. Theorem 36
in~\cite{AIM14arx}), which gives only the weaker lower bound of 
$n^{\text{poly}(\eps) \log(n)^{1 - o(1)}}$.}. We shall refer to this as the BKW
result.
 
It is worth noting that the Rubinstein's hardness result~\cite{R16} almost makes
this result redundant. If one is willing to accept that PETH is true, which is a
stronger conjecture than ETH, 
then Rubinstein's result says that for small
$\eps$ we require quasi-polynomial time to find \emph{any} $\eps$-NE, which
obviously implies that the same lower bound applies to $\eps$-NE $\delta$-SW for any
$\delta$.


\noindent \textbf{Our results.}
Our first result is a lower bound for the problem of finding $\eps$-NE
$\delta$-SW. 
We show that, assuming ETH, that there exists a small constant $\delta$ such
that the problem $\left(\frac{1 - 4 g  \cdot
\delta}{8}\right)$-NE $\left(\frac{g \cdot \delta}{4}\right)$-SW requires
$n^{\mathrm{\widetilde\Omega}( \log n)}$ time%
\footnote{Here $\mathrm{\widetilde \Omega}(\log n)$ means $\mathrm{\Omega}(\frac{\log n}{(\log
\log n)^c})$ for some constant $c$.},
where $g = \frac{1}{138}$.

To understand this result, let us compare it to the BKW result. First, observe
that as $\delta$ gets smaller, the $\eps$ in our $\eps$-NE gets larger, whereas
in the BKW result, $\eps$ get smaller. Asymptotically,
our $\eps$ approaches $1/8$. Moreover, since $\delta \le 1$, our lower bound
applies to all $\eps$-NE with $\eps \le \frac{1 - 4g}{8} \approx 0.1214$. This
is orders of magnitude larger than the inapproximability bound given by
Rubinstein's hardness result, and so is not made redundant by that result. In
short, our hardness result is about the hardness of obtaining good social
welfare, rather than the hardness of simply finding an approximate equilibrium.

Secondly, when compared to the BKW result, we obtain a slightly better lower
bound. The exponent in their lower bound is logarithmic only in the limit, while
ours is always logarithmic. 

The second set of results in this paper show that,
once we have our lower bound on the problem of finding $\eps$-NE $\delta$-SW, we
use it to prove lower bounds for other problems regarding constrained
approximate NEs and WSNEs.  Table~\ref{tab:problems} gives a list of the
problems that we consider. For each one, we provide a reduction from $\eps$-NE
$\delta$-SW to that problem.  Ultimately, we prove that if ETH is true, then
for every $\eps < \frac{1}{8}$ finding an $\eps$-NE with the given property
in an $n \times n$ bimatrix game 
requires $n^{\mathrm{\widetilde \Omega}(\log n)}$ time.



\probc{probc:largep}
\probc{probc:restrictedweak}
\probc{probc:difcomp}
\probc{probc:maxprob}
\probc{probc:smalltp}
\probc{probc:smallp}
\probc{probc:totalmaxsupport}
\probc{probc:minmaxsupport}
\probc{probc:maxsupport}
\probc{probc:restricting}

\renewcommand*{\arraystretch}{1.6}
\setlength{\tabcolsep}{7pt}
\begin{table}[htpb!]
\centering
\resizebox{\textwidth}{!}{
\begin{tabular}{|p{0.45\textwidth}|p{0.55\textwidth}|}
\hline
\textbf{Problem description} & \textbf{Problem definition} \\
\hline
\hline
Problem~\ref{probc:largep}: Large payoffs $u \in (0,1]$
&
Is there an \eps-NE $(\xbf, \ybf)$ such that
 $\min(\xbf^TR\ybf,\xbf^TC\ybf) \geq u$?\\
\hline
Problem~\ref{probc:restrictedweak}: Restricted support $S \subset [n]$
&
Is there an \eps-NE $(\xbf, \ybf)$ with $\supp(\xbf) \subseteq S$? \\
\hline
Problem~\ref{probc:difcomp}: 
Two \eps-NE $d \in (0,1]$ apart in Total Variation (TV) distance
&
Are there two \eps-NE with TV distance $\ge d$? \\
\hline
Problem~\ref{probc:maxprob}: Small largest probability $p \in (0,1)$
&
Is there an \eps-NE $(\xbf, \ybf)$ with $\max_i\xbf_i \leq p$?\\
\hline
Problem~\ref{probc:smalltp}: Small total payoff $v \in [0,2)$
&
Is there an \eps-NE $(\xbf, \ybf)$ such that 
$\xbf^TR\ybf + \xbf^TC\ybf \leq v$? \\
\hline
Problem~\ref{probc:smallp}: Small payoff $u \in [0,1)$
&
Is there an \eps-NE $(\xbf, \ybf)$ such that $\xbf^TR\ybf \leq u$?\\
\hline

Problem~\ref{probc:totalmaxsupport}: Large total 
support size $k \in [n]$
&
Is there an \eps-WSNE $(\xbf, \ybf)$ such that 
$|\supp(\xbf)| + |\supp(\ybf)| \geq 2k$?\\
\hline
Problem~\ref{probc:minmaxsupport}: Large smallest
support size 
$k \in [n]$
&
Is there an \eps-WSNE $(\xbf, \ybf)$ such that 
$\min\{|\supp(\xbf)|, |\supp(\ybf)|\} \geq k$?\\
\hline
Problem~\ref{probc:maxsupport}: Large support size $k \in [n]$ 
&
Is there an \eps-WSNE $(\xbf, \ybf)$ such that $|\supp(\xbf)| \geq k$?\\
\hline
Problem~\ref{probc:restricting}: Restricted support
$S_R \subseteq [n]$
&
Is there an \eps-WSNE $(\xbf, \ybf)$ with $S_R \subseteq \supp(\xbf)$? \\
\hline
\end{tabular}
}
\medskip
\caption{The decision problems that we consider. 
All of them take as input a bimatrix game $(R,C)$ and a quality of approximation 
$\eps \in (0,1)$. 
Problems~\ref{probc:largep}~-~\ref{probc:maxprob} relate to \eps-NE, and 
Problems~\ref{probc:totalmaxsupport}~-~\ref{probc:restricting} relate to
\eps-WSNE.
}
\label{tab:problems} 
\end{table}

\paragraph{\bf Techniques}

At a high level, the proof of our first result is similar in spirit 
to the proof of the BKW result. They reduce from the problem of
approximating the value of a \emph{free game}. Aaronson, Impagliazzo, and
Moshkovitz showed quasi-polynomial lower bounds for this problem assuming
ETH~\cite{AIM14}.  
A free game is played between to players named \mone and \mtwo, and a referee
named \art. The BKW result creates a bimatrix game that simulates the free game,
where the two players take the role of \mone and \mtwo, while \art is simulated
using a zero-sum game. 

Our result will also be proved by producing a bimatrix game that simulates a
free game. However, there are a number of key differences that allow us to prove
the stronger lower bounds described above. The first key difference is that we
use a different zero-sum game to simulate \art.
Our zero-sum game is inspired by the one used by Feder, Nazerzadeh, and
Saberi~\cite{FNS07}. The advantage of this construction is that it is capable of
ensuring that play distributions that a very close to uniform in all approximate
Nash equilibria, which in turn gives us a very accurate simulation of \art. 

The downside of this zero-sum game is that it requires $2^n$ rows to force the
column player to mix close to uniformly over $n$ rows. \art is required to pick
two \emph{questions} uniformly from a set of possible questions. The free games
provided by Aaronson, Impagliazzo, and Moshkovitz have question sets of linear
size, so if we reduce directly from these games, we would end up with an
exponentially sized bimatrix game. The conference version of this
paper~\cite{DFS16} resolved the issue by using a sub-sampling lemma, also proved
by Aaronson, Impagliazzo, and Moshkovitz, that produces a free game with
logarithmically sized question sets. This allowed us to produce polynomially
sized bimatrix games, but at the cost of needing randomization to implement the
reduction, and so the result depended on the randomized version of the ETH. 

In this version, we show that we are able to assume only the ETH by using a
stronger result that was discovered by Babichenko, Papadimitriou, and
Rubinstein~\cite{BPR16}. Their results imply that approximating the value of a
free game requires quasipolynomial time even when the size of the question sets
is logarithmic in the game size. They do not explicitly formulate this result,
but it is clearly implied by their techniques. For the sake of completeness, we
provide an exposition of their ideas in Section~\ref{sec:bpr}.

The second main difference between our result and the BKW result is that we use
a different starting point. The BKW result uses the PCP theorem of Moshkovitz
and Raz~\cite{MR10}, which provides a completeness/soundness gap of $1$ vs
$\delta$ for arbitrarily small constant $\delta$ in the label cover problem. The
use of this powerful PCP theorem is necessary, as their proof relies on the
large completeness/soundness gap produced by that theorem. This choice of PCP
theorem directly impacts the running time lower bound that they produce, as the
$(\log n)^{1 - o(1)}$ term in the exponent arises from the blowup of $n^{1 +
o(1)}$ from the PCP theorem.

In contrast to this, our stronger simulation of \art allows us to use the PCP
theorem of Dinur~\cite{Dinur07} as our starting point. This PCP theorem only
involves a blowup of $n \polylog(n)$, which directly leads to the improved
$\widetilde{\Omega}(\log n)$ exponent in our lower bound. The improved blowup of
the PCP theorem comes at the cost
of providing a completeness/soundness gap of only 1 vs $1 - \eps$ for $\eps <
\frac{1}{8}$, but our simulation is strong enough to deal with this. It is also
worth noting that if a PCP theorem with a constant completeness/soundness gap
and linear blow up is devised in the future, then the exponent in our lower bound will improve
to $\Omega(\log n)$.

One final point of comparison is the size of the payoffs used in our simulation.
The zero-sum games that we use have payoffs in the range $(-4, 4)$, which
directly leads to the $\frac{1-4g\cdot \delta}{8}$ bound on the quality of
approximation. In contrast to this, the zero-sum games used by the BKW result
have payoffs of size $O(\frac{1}{\eps})$, which ultimately means that their
lower bound only applies to the problem $\eps$-NE $\eps$-SW.





\paragraph{\bf Other related work}

The only positive result for finding $\eps$-NE with good social welfare that we
are aware of was given by Czumaj,  Fasoulakis,  and Jurdzi\'nski~\cite{CFJ15,CFJ16}.
In~\cite{CFJ15}, they showed that if there is a polynomial-time algorithm for finding an
$\eps$-NE, then for all $\eps' > \eps$ there is also a polynomial-time algorithm
for finding an $\eps'$-NE that is within a constant multiplicative approximation
of the best social welfare.  They also give further results for the case where
$\eps > \frac{1}{2}$. In~\cite{CFJ16} they derived polynomial-time algorithms that 
compute \eps-NE for $\eps \geq \frac{3-\sqrt{5}}{2}$ that
approximate the quality of plutocratic and egalitarian Nash equilibria
to various degrees.

\section{Preliminaries}
\label{sec:pre}

Throughout the paper, we use $[n]$ to denote the set of integers $\{1, 2, \dots,
n\}$. An $n \times n$ bimatrix game is a pair $(R,C)$ of two $n \times n$
matrices: $R$ gives payoffs for the \emph{row} player and $C$ gives the payoffs
for the \emph{column} player. 


Each player has $n$ \emph{pure} strategies. To play the
game, both players simultaneously select a pure strategy: the row player selects
a row $i \in [n]$, and the column player selects a column $j \in [n]$. The row
player then receives payoff $R_{i,j}$, and the column player receives payoff
$C_{i,j}$.

A \emph{mixed strategy} is a probability distribution over $[n]$. We denote a
mixed strategy for the row player as a vector \xbf of length $n$, such that
$\xbf_i$ is the probability that the row player assigns to pure strategy $i$.
A mixed strategy of the column player is a vector \ybf of length $n$, with the
same interpretation. 
If \xbf and~\ybf are mixed strategies for the row and the column player,
respectively, then we call $(\xbf, \ybf)$ a \emph{mixed strategy profile}.
The expected payoff for the row player under strategy profile $(\xbf, \ybf)$
is given by $\xbf^T R \ybf$ and for the column player by $\xbf^T C
\ybf$.  
We denote the \emph{support} of a strategy $\xbf$ as $\supp(\xbf)$,
which gives the set of pure strategies~$i$ such that $\xbf_i > 0$.

\paragraph{\bf Nash equilibria}

Let $\ybf$ be a mixed strategy for the column player. The set of \emph{pure
best responses} against $\ybf$ for the row player is the set of pure
strategies that maximize the payoff against $\ybf$. More formally, a pure
strategy $i \in [n]$ is a best response against $\ybf$ if, for all pure
strategies $i' \in [n]$ we have: $\sum_{j \in \protect [n]} \ybf_j \cdot R_{i,
j} \ge \sum_{j \in \protect [n]} \ybf_j \cdot R_{i', j}$. Column player best
responses are defined analogously.

A mixed strategy profile $(\xbf, \ybf)$ is a \emph{mixed Nash equilibrium}
if every pure strategy in $\supp(\xbf)$ is a best response against~$\ybf$,
and every pure strategy in $\supp(\ybf)$ is a best response against~$\xbf$.
Nash~\cite{N51} showed that every bimatrix game has a mixed Nash equilibrium.
Observe that in a Nash equilibrium, each player's expected payoff is equal to
their best response payoff.

\paragraph{\bf Approximate Equilibria}

There are two commonly studied notions of approximate equilibrium, and we
consider both of them in this paper. The first notion is that of an
\emph{$\epsilon$-approximate Nash equilibrium} ($\epsilon$-NE), which weakens the
requirement that a player's expected payoff should be equal to their best
response payoff. Formally, given a strategy profile $(\xbf, \ybf)$, we
define the \emph{regret} suffered by the row player to be the difference between
the best response payoff and the actual payoff:
$\max_{i\in [n]} \big((R\cdot y)_i\big) - \xbf^T \cdot R \cdot \ybf.$
Regret for the column player is defined analogously. We have that $(\xbf,
\ybf)$ is an $\epsilon$-NE if and only if both players have regret less than
or equal to $\epsilon$.

The other notion is that of an $\epsilon$-approximate-well-supported equilibrium
($\epsilon$-WSNE), which weakens the requirement that players only place
probability on best response strategies. 
We say that a pure strategy $j \in [n]$ of the row player is an
$\epsilon$-best-response against $\ybf$ if:
\begin{equation*}
\max_{i \in [n]}\big((R\cdot y)_i\big) - (R\cdot y)_j \le \epsilon.
\end{equation*}
An $\epsilon$-WSNE requires that both players only place probability on
$\epsilon$-best-responses. Formally,
the row player's \emph{pure strategy regret} under $(\xbf, \ybf)$ is defined
to be:
$\max_{i \in [n]}\big((R\cdot y)_i\big) -
\min_{i \in \supp(\xbf)}\big((R\cdot y)_i\big).$
Pure strategy regret for the column player is defined analogously. A strategy
profile $(\xbf, \ybf)$ is an $\epsilon$-WSNE if both players have pure
strategy regret less than or equal to $\epsilon$.

Since approximate Nash equilibria use an additive notion of approximation, it is
standard practice to rescale the input game so that all payoffs lie in the range
$[0, 1]$, which allows us to compare different results on this topic. For the
most part, we follow this convention. However, for our result in
Section~\ref{sec:bday}, we will construct a game whose payoffs do not lie in
$[0,1]$. In order to simplify the proof, we will prove results about approximate
Nash equilibria in the unscaled game, and then rescale the game to $[0, 1]$ at
the very end. To avoid confusion, we will refer to an $\eps$-approximate Nash
equilibrium in this game as an $\eps$-UNE, to mark that it is an additive
approximation in an unscaled game.

\paragraph{\bf Two-prover games}

A two-prover game is defined as follows.

\begin{definition}[Two-prover game]
A two-prover game $\mathcal{T}$ is defined by a tuple $(X, Y, A, B, \mathcal{D}, V)$
where
$X$ and $Y$ are finite sets of \emph{questions},
$A$ and $B$ are finite sets of \emph{answers}, 
\dcal is a probability distribution defined over $X \times Y$,
and $V$ is a \emph{verification function} of the form $V: X \times Y \times A \times B \rightarrow \{0,1\}$.
\end{definition}

The game is a co-operative and played between two players, who are called \mone
and \mtwo, and an adjudicator called \art. At the start of the game, \art
chooses a question pair $(x,y) \in X \times Y$ randomly according to~$\dcal$. He
then sends $x$ to \mone and $y$ to $\mtwo$. Crucially, \mone does not know the
question sent to \mtwo and vice versa. Having received $x$, \mone then chooses
an answer from $A$ and sends it back to \art. \mtwo similarly picks an answer
from $B$ and returns it to \art. \art then computes $p = V(x, y, a, b)$ and
awards payoff $p$ to both players. The size of the game, denoted $|\mathcal{T}| = | X
\times Y \times A \times B|$ is the total number of entries needed to represent
$V$ as a table.

A \emph{strategy} for \mone is a function $a: X \rightarrow A$ that gives an
answer for every possible question, and likewise a strategy for \mtwo is a
function $b: Y \rightarrow B$. We define $S_i$ to be the set of all strategies
for \meye. The \emph{payoff} of the game under a pair of strategies $(s_1, s_2)
\in S_1 \times S_2$ is denoted as 
\begin{equation*}
p(\mathcal{T}, s_1, s_2) = E_{(x,y) \sim \dcal}[V(x, y, s_1(x), s_2(y))].
\end{equation*}

The \emph{value} of the game, denoted $\omega(\tcal)$, is the maximum expected
payoff to the Merlins when they play optimally: 
\begin{equation*}
\omega(\mathcal{T}) = \max_{s_1 \in S_1} \max_{s_2 \in S_2} p(\mathcal{T}, s_1,
s_2).
\end{equation*}

\paragraph{\bf Free games}

A two-prover game is called a \emph{free game} if the probability 
distribution~\dcal is the
uniform distribution \ucal over $X \times Y$. In particular, this means that
there is no correlation between the question sent to \mone and the question sent
to \mtwo. We are interested in the problem of approximating the value of a free
game within an additive error of $\delta$.
%
\bigskip
\begin{tcolorbox}[title=\freegame]
Input: A free game $\mathcal{T}$ and a constant $\delta > 0$.

\medskip
Output: A value $p$ such that $| \; \omega(\tcal) - p \; | \le \delta$.
\end{tcolorbox}

\section{Hardness of approximating free games}
\label{sec:bpr}

The \emph{exponential time hypothesis (ETH)} is the conjecture that any
deterministic algorithm for solving 3SAT requires $2^{\mathrm{\Omega}(n)}$ time.
Aaronson, Impagliazzo, and Moshkovitz have shown that,
if ETH holds, then there exists a small constant $\eps > 0$ such that
approximating the value of a free game within an additive error of $\eps$
requires quasi-polynomial time. However, their result is not suitable for our
purposes, because it produces a free game in which the question and answer sets
have the same size, and to prove our result, we will require that the question sets have logarithmic size when compared to the answer sets. 

The conference version of this paper~\cite{DFS16} solved this issue by using a
sub-sampling lemma, also proved by Aaronson, Impagliazzo, and Moshkovitz, which
shows that if we randomly choose logarithmically many questions from the
original game, the value of the resulting sub-game is close the value of the
original. However, this comes at the cost of needing randomness in the
reduction, and so our result depended on the truth of the \emph{randomized} ETH,
which is a stronger conjecture.

In this exposition, we will instead use a technique of Babichenko,
Papadimitriou, and Rubinstein~\cite{BPR16}, which allows us to produce a free
game with a logarithmic size question set in a deterministic way. 
The result that we need is a clear consequence of their ideas, but is not
explicitly formulated in their paper. For the sake of completeness, 
in the rest of
this section we provide our own exposition of their ideas.


\paragraph{\bf The PCP theorem}

The starting point of the result will be a 3SAT instance~$\phi$. We say that the
size of a formula~$\phi$ is the number of variables and clauses in the formula.
We define $\sat(\phi) \in [0, 1]$ to be the maximum fraction of clauses that can
be satisfied in $\phi$. The first step is to apply a PCP theorem.

\begin{theorem}[Dinur's PCP Theorem~\cite{Dinur07}]
\label{thm:dinur}
Given any 3SAT instance $\phi$ of size $n$, and a constant $\epsilon$ in the
range $0 < \epsilon < \frac{1}{8}$, we can produce in polynomial time a 3SAT
instance $\psi$ where:
\begin{itemize}
\item The size of $\psi$ is $n \cdot \polylog(n)$.
\item Every clause of $\psi$ contains exactly 3 variables and every variable is
contained in at most $d$ clauses, where $d$ is a constant.
\item If $\sat(\phi) = 1$, then $\sat(\psi) = 1$.
\item If $\sat(\phi) < 1$, then $\sat(\psi) < 1 - \epsilon$.
\end{itemize}
\end{theorem}

After applying the PCP theorem given above, we then directly construct a free
game. Observe that a 3SAT formula can be viewed as a bipartite graph in which
the vertices are variables and clauses, and there is an edge between a variable
$x_i$ and a clause $C_j$ if and only if $x_i$ is appears in $C_j$. In
particular, the 3SAT formulas produced by Theorem~\ref{thm:dinur} correspond to 
bipartite graphs with constant degree, since each clause has degree at most 3,
and each variable has degree at most $d$.

The first step is to apply the following lemma, which allows us to partition the
vertices of this bipartite graph. The lemma and proof are essentially
identical to~\cite[Lemma 6]{BPR16}, although we generalise the formulation
slightly, because the original lemma requires that the two sides of the graph
have exactly the same number of nodes and that the graph is $d$-regular.

\begin{lemma}[\cite{BPR16}]
\label{lem:split}
Let $(V, E)$ be a bipartite graph with $|V| = n$, where
$V = U \cup W$ are the two sides of the graph, and where each node has degree at
most $d$. 
Suppose that $U$ and $W$ both have a constant fraction of the vertices, 
and hence $|U| = c_1 \cdot n$ and $|W| = c_2 = (1-c_1) \cdot n$ for some constants 
$c_1 < 1$ and $c_2 < 1$.
We can efficiently  find a partition $S_1, S_2,
\dots, S_{\sqrt n}$ of $U$ and a partition $T_1, T_2, \dots, T_{\sqrt n}$ of $W$
such that
each set has size at most $2 \sqrt n$, and for all $i$ and~$j$ we have
\begin{equation*}
|(S_i \times T_j) \cap E| \le 2 d^2.
\end{equation*}
\end{lemma}
\begin{proof}
The algorithm is as follows. First we arbitrarily split $U$ into $\sqrt n$ many
sets $S_1, S_2, \dots, S_{\sqrt n}$, and so each set $S_i$ has size $c_1 
\sqrt n < 2 \sqrt n$. Then we iteratively construct the partition of $W$ into sets $T_1,
T_2, \dots, T_{\sqrt n}$ in the following
way. We initialize each set $T_j$ to be the empty set. In each iteration, we
pick a vertex of $w \in W$ that has not already been assigned to a set. We
find a set $T_j$ such that $|T_j| \le 2 \cdot \sqrt{n}$, and such that for
all $i$ we have $|(S_i \times T_j) \cap E| \le 2 \cdot d^2$. We assign
$w$ to $T_j$ and repeat.

Obviously, for the algorithm to be correct, we must prove that for each vertex
$w$ that is considered, there does exist a set $T_j$ that satisfies the required
constraints. For this, we rely on the following two properties.
\begin{itemize}
\item The average number of
vertices in a set $T_j$ is at most $c_2 \sqrt n < \sqrt n$, and so by Markov's
inequality strictly less than half the sets can have size more than $2 \sqrt n$,
and so we lose strictly less than half the sets $T_j$ to the size constraint.

\item Since each vertex has degree at most $d$, the graph has at most $dn$
edges, and so the average number of edges between each pair of sets $S_i$ and $T_j$ is
$dn/(\sqrt{n} \cdot \sqrt{n}) = d$. Again, using Markov's inequality we can
conclude that there are at most $1/2d$ pairs of sets $S_i$
and $T_j$ that have more than $2d^2$ edges between them. Hence, even in the
worst case, we can lose at most $1/2d$ sets $T_j$ to the edge constraints.
\end{itemize}
So, we lose strictly less than half the sets to the size constraints, and $1/2d
\le 1/2$ the sets to the edge constraints. Hence, by the union bound, we have
shown that there is at least one set $T_j$ that satisfies both constraints
simultaneously.
\qed
\end{proof}

\paragraph{\bf A free game}
Note that Lemma~\ref{lem:split} can be applied to the 3SAT formula that arises
from Dinur's PCP theorem, because the number of variables and number of
constraints are both a constant fraction of the number of nodes in the
associated bipartite graph, and because each vertex has either has degree $d$ or
degree $3$. We use this to construct the following free game, which is highly
reminiscent of the clause variable game given by Aaronson, Impagliazzo, and
Moshkovitz~\cite{AIM14}.

\begin{definition}
Given a 3SAT formula $\phi$ of size $n$, we define a free game $\mathcal{F}_\phi$ in the
following way.
\begin{enumerate}
\item \art begins by applying Dinur's PCP theorem to $\phi$ to obtain a formula
$\psi$ of size $N = n \polylog(n)$,
and then uses  Lemma~\ref{lem:split} to split the variables of $\psi$ into sets
$S_1, S_2, \dots, S_{\sqrt{n}}$ and the clauses of $\psi$ into sets $T_1, T_2,
\dots, T_{\sqrt{N}}$.
\item \art picks an index $i$ uniformly at random from $[\sqrt{N}]$, 
and independently an index $j$ uniformly at random from $[\sqrt{N}]$. He
sends $S_i$ to \mone and $T_j$ to \mtwo.
\item \mone responds by giving a truth assignment to every variable in $S_i$,
and \mtwo responds by giving a truth assignment to every variable that is
involved with a clause in $T_j$.
\item \art awards the Merlins payoff $1$ if and only if both of the following
conditions hold.
\begin{itemize}
\item \mtwo returns an
assignment that satisfies all clauses in $T_j$.
\item For every variable $v$ that appears in $S_i$ and some clause of $T_j$, the
assignment to $v$ given by $\mone$ agrees with the assignment to $v$ given by
$\mtwo$. Note that this condition is always satisfied when $S_i$ and $T_j$ share
no variables.
\end{itemize}
\art awards payoff $0$ otherwise.
\end{enumerate}
\end{definition}

If $n$ is the size of $\phi$, then
when we write $\mathcal{F}_\phi$ down as a free game $(X, Y, A, B,
\mathcal{D}, V)$, the number of questions in the sets $X$ and $Y$ is $\sqrt{n
\polylog(n)}$, and the number of answers in $A$ and $B$ is $2^{2\sqrt{n
\polylog(n)}}$, where the extra $\polylog(n)$ factor arises due to the
application of the PCP theorem.

The following lemma shows that if $\phi$ is unsatisfiable, then the value of
this free game is bounded away from $1$. Again, the ideas used to prove this
lemma are clearly evident in the work of Babichenko, Papadimitriou, and
Rubinstein~\cite{BPR16}.

\begin{lemma}[\cite{BPR16}]
\label{lem:csgap}
If $\phi$ is satisfiable then $\omega(\mathcal{F}_\phi) = 1$.
If $\phi$ is unsatisfiable then $\omega(\mathcal{F}_\phi) \le 1 - \eps/2d$.
\end{lemma}
\begin{proof}
The case where $\sat(\phi) = 1$ is straightforward. Since there exists a
satisfying assignment for $\phi$, there also exists a satisfying assignment 
for $\psi$. If the two Merlins play according to this satisfying assignment,
then they obviously achieve an expected payoff of $1$.

For the other claim, first observe that we can assume that both Merlins play
deterministic strategies, since the game is co-operative, and therefore nothing
can be gained through randomization. So, let $s_1$ be a strategy for \mone.
Observe that since $S_1, S_2, \dots, S_{\sqrt{N}}$ partition the variables of
$\psi$, we have that $s_1$ yields an assignment to the variables of $\psi$. 

Let us fix an arbitrary deterministic strategy $s_1$ for \mone. We have that the
payoff to \mtwo for an individual question $T_j$ can be computed as follows:
\begin{itemize}
\item For every set $S_i$ for which there are no edges between the variables in
$S_i$ and $T_j$, \mtwo gets payoff $1$ ``for free.''
\item Otherwise, \mtwo gets payoff $1$ only if the assignments to the clauses
agree with the assignment implied by $s_1$.
\end{itemize}
From this, we can see that when \mone plays $s_1$, \mtwo can maximize
his payoff by playing the strategy that agrees everywhere with the assignment
chosen by \mone. So let $s_2$ denote this strategy.

Since $\phi$ is unsatisfiable, the PCP theorem tells us that $\sat(\psi) < 1 -
\eps$. 
Thus, there are at least $\eps d N$ clauses that are not satisfied when
$s_1$ is played against $s_2$. Since Lemma~\ref{lem:split} ensures that the
maximum number of edges between two sets is $2d^2$, there must therefore be at
least $\eps dN/2d^2 = \eps N/2d$ pairs of sets that give payoff $0$ to the
Merlins under $s_1$ and $s_2$.
Since there are exactly $N$ pairs of sets in total, this means that the expected
payoff to the Merlins is bounded by $1 - \eps/2d$.
\qed
\end{proof}

Finally, we can formulate the lower bound that we will use 
in this paper. The proof is the same as the one given
in~\cite{AIM14}, but we use the free game $\mathcal{F}_\phi$, rather than the
construction originally given in that paper.

\begin{theorem}[\cite{BPR16}]
\label{thm:qplower}
Assuming ETH, there is a small constant $\delta$ below which the problem
\freegame cannot be solved faster than $N^{\mathrm{\widetilde O}(\log N)}$, even
when the question sets have size $\log N$.
\end{theorem}
\begin{proof}
Lemma~\ref{lem:csgap} implies that if we can approximate the value of 
$\mathcal{F}_\phi$ with an additive error of less than $\eps/2d$, then we can
solve the satisfiability problem for~$\phi$.

Assume, for the sake of contradiction, that there exists an algorithm that can
solve the \freegame problem in time $N^{o\left(\frac{\log
N}{(\log \log N)^c}\right)}$ for some constant $c$ that will be fixed later. 
Observe that the free game~$\mathcal{F}_\phi$ has size $N = O(2^{\sqrt{n
\polylog(n)}})$, and so the hypothesized algorithm would run in time:
\begin{align*}
& \exp\left( o \left( \frac{ \log^2 N } {(\log\log(N))^c} \right) \right) 
	= \exp\left( o \left( \frac{ n \polylog(n) } {(\log( \sqrt{n}
\polylog(n)))^c} \right) \right).
\end{align*}
If we set $c$ to be greater than the degree of the polynomial in the
$\polylog(n)$ from the numerator, then we can conclude that the running time
would be $2^{o(n)}$, which would violate the ETH.
\qed
\end{proof}





\section{Hardness of approximating social welfare}
\label{sec:bday}

\paragraph{\bf Overview}
In this section, we study the following \emph{social welfare} problem for a bimatrix
game $\gcal = (R,C)$.
The \emph{social welfare} of a strategy profile $(\xbf, \ybf)$ is denoted by
$\sw(\xbf, \ybf)$ and is defined to
be $\xbf^TR\ybf + \xbf^TC\ybf$. Given an $\eps \ge 0$, we define the set of all
$\eps$ equilibria as 
\begin{equation*}
\eeq = \{(\xbf, \ybf) \; : \; (\xbf, \ybf) \text{ is an } \eps\text{-NE}\}.
\end{equation*}
Then, we define the \emph{best social welfare} achievable by an $\eps$-NE in
$\mathcal{G}$ as 
\begin{equation*}
\bsw(\gcal, \eps) = \max\{\sw(\xbf, \ybf) \; : \; (\xbf, \ybf) \in \eeq\}.
\end{equation*}
Using these definitions we now define the main problem that we consider:

\bigskip
\begin{tcolorbox}[title=$\eps$-NE $\delta$-SW]
Input: A bimatrix game $\mathcal{G}$, and two constants $\eps, \delta > 0$.
\smallskip

Output: An $\eps$-NE $(\xbf,\ybf)$ s.t. $\sw(\xbf, \ybf)$ is within $\delta$ of $\bsw(\gcal, \eps)$. 
\end{tcolorbox}

\noindent
We show a lower bound for this problem by reducing from \freegame. 
Let~$\mathcal{F}$ be a free game of size $n$ from the family of free
games that were used to prove Theorem~\ref{thm:qplower} (from now on we will
drop the subscript $\phi$, since the exact construction of $\mathcal{F}$ is not
relevant to us.) 
We have that either
$\omega(\mathcal{F}) = 1$ or $\omega(\mathcal{F}) < 1 - \delta$ for some fixed
constant $\delta$, and that it is hard to determine which of these is the case.
We will construct a game $\mathcal{G}$ such that for $\eps = 1 - 4g \cdot
\delta$, where $g < \frac{5}{12}$ is a fixed constant that we will define at the
end of the proof, we have the following properties.
\begin{itemize}
\item (\textbf{Completeness}) If $\omega(\mathcal{F}) = 1$, then the unscaled
$\bsw(\mathcal{G}, \eps) = 2$.
\item (\textbf{Soundness}) If $\omega(\mathcal{F}) < 1 - \delta$, then the unscaled
$\bsw(\mathcal{G}, \eps) < 2(1 - g \cdot \delta)$.
\end{itemize}
This will allow us to prove our lower bound using Theorem~\ref{thm:qplower}.

\subsection{The construction}

%
We use $\fcal$ to construct a bimatrix game, which we will denote as
$\mathcal{G}$ throughout the rest of this section. The game is built out of four
subgames, which are arranged and defined as follows.

\begin{center}
{\small
\bimatrixgame{3.5mm}{2}{2}{I}{II}%
{{ }{ }}%
{{ }{ }}%
{
\payoffpairs{1}{{$R$}{$-D_2$}}{{$C$}{$D_2$}}
\payoffpairs{2}{{$D_1$}{$0$}}{{$-D_1$}{$0$}}
}
	}
\end{center}
\begin{itemize}

\item The game $(R, C)$ is built from $\fcal$ in the following way.
Each row of the game corresponds to a pair $(x,a) \in 
X \times A$ and each column corresponds to a pair $(y,b)
\in Y \times B$. Since all free games are cooperative,  the payoff for each
strategy pair $(x,a),(y,b)$ is defined to be
$R_{(x,a),(y,b)} = C_{(x,a),(y,b)} = V(x, y, a(x), b(y)).$

\item The game $(D_1, -D_1)$ is a zero-sum game. The game is a slightly modified
version of a game devised by Feder, Nazerzadeh, and Saberi~\cite{FNS07}. Let $H$
be the set of all functions of the form $f : Y \rightarrow \{0, 1\}$ such that
$f(y) = 1$ for exactly half%
\footnote{If $|Y|$ is not even, then we can create a new free game in which
each question in $|Y|$ appears twice. This will not change the value of the free
game.}
of the elements $y \in Y$. The game has $|Y \times
B|$ columns and $|H|$ rows. For all $f \in H$ and all $(y, b) \in Y$ the payoffs
are
\begin{equation*}
\left(D_1\right)_{f, (y, b)} = 
\begin{cases}
\frac{4}{1 + 4g \cdot \delta} & \text{if $f(y) = 1$,} \\
0 & \text{otherwise.}
\end{cases}
\end{equation*}

\item The game $(-D_2, D_2)$ is built in the same way as the game $(D_1, -D_1)$,
but with the roles of the players swapped. That is, each column of $(-D_2, D_2)$
corresponds to a function that picks half of the elements of $X$.

\item The game $(0, 0)$ is a game in which both players have zero matrices.

\end{itemize}

Observe that the size of $(R, C)$ is the same as the size of $\mathcal{F}$. The
game $(D_1, -D_1)$ has the same number of columns as~$C$, and the number of rows
is at most $2^{|Y|} \le 2^{ O(\log |\mathcal{F}|)} = |\mathcal{F}|^{O(1)}$,
where we are crucially using the fact that Theorem~\ref{thm:qplower} allows us
to assume that the size of $Y$ is $O(\log |\mathcal{F}|)$. By the same
reasoning, the number of columns in $(-D_2, D_2)$ is at most
$|\mathcal{F}|^{O(1)}$. Thus, the size of $\mathcal{G}$ is
$|\mathcal{F}|^{O(1)}$, and so this reduction is polynomial.

\subsection{Completeness}

To prove completeness, it suffices to show that, if $\omega(\fcal) = 1$, then
there exists a 
$(1 - 4g \cdot \delta)$-UNE of $\mathcal{G}$ that has social welfare $2$. To do
this, assume that~$\omega(\fcal) = 1$, and take a pair of optimal strategies
$(s_1, s_2)$ for $\fcal$ and turn them
into strategies for the players in~$\mathcal{G}$. More precisely, the row player
will place probability $\frac{1}{|X|}$ on each answer chosen by $s_1$, and the
column player  will place probability $\frac{1}{|Y|}$ on each answer chosen by
$s_2$. By construction, this gives both players payoff $1$, and hence the social
welfare is $2$. The harder part is to show that this is an approximate
equilibrium, and in particular, that neither player can gain by playing a
strategy in $(D_1, -D_1)$ or $(-D_2, D_2)$. We prove this in the following
lemma.

\begin{lemma}
\label{lem:complete}
If $\omega(\fcal) = 1$, then there exists a $(1 - 4g \cdot \delta)$-UNE $(\xbf, \ybf)$
of $\mathcal{G}$ with $\sw(\xbf, \ybf) = 2$.
\end{lemma}

\begin{proof}
Let $(s_1, s_2) \in S_1 \times S_2$ be a pair of optimal strategies for \mone
and \mtwo in \fcal. For each $(x, a) \in X \times A$ and each $(y, b) \in Y
\times B$, we define
\begin{align*}
\xbf(x, a) = \begin{cases}
\frac{1}{|X|} & \text{if $s_1(x) = a$.} \\
0 & \text{otherwise.}
\end{cases}
&&
\ybf(y, b) = \begin{cases}
\frac{1}{|Y|} & \text{if $s_2(y) = b$.} \\
0 & \text{otherwise.}
\end{cases}
\end{align*}
Clearly, by construction, we have that the payoff to the row player under
$(\xbf, \ybf)$ is equal to $p(\fcal, s_1, s_2) = 1$, and therefore $(\xbf, \ybf)$
has social welfare~$2$.

On the other hand, we must prove that $(\xbf, \ybf)$ is a 
$(1-4g\cdot\delta)$-UNE. To do
so, we will show that neither player has a deviation that increases their payoff
by more than $(1-4g\cdot\delta)$. We will show
this for the row player; the proof for the column player is symmetric.
There are two types of row to consider.
\begin{itemize}
\item First suppose that $r$ is a row in the sub-game $(R, C)$. We claim that
the payoff of $r$ is at most $1$. This is because the maximum payoff in $R$ 
is~$1$, while the maximum payoff in $-D_2$ is $0$. Since the row player already
obtains payoff $1$ in $(\xbf, \ybf)$, row $r$ cannot be a profitable deviation.

\item Next suppose that $r$ is a row in the sub-game $(D_1, -D_1)$. Since we 
have
$\sum_{b \in B} \ybf(y, b) = \frac{1}{|Y|}$ for every question $y$, we have that 
all rows in~$D_1$ have the same payoff. This payoff is
\begin{equation*}
\frac{1}{2} \cdot \left( \frac{4}{1 + 4\cdot \delta} \right)  = 
\frac{2}{1 + 4g\cdot \delta} = 2 - \frac{8g \cdot \delta}{1 + 4g \cdot \delta}\ .
\end{equation*}
Since $\delta \le 1$ and $g \le \frac{1}{4}$ we have
\begin{equation*}
\frac{8}{1 + 4g \cdot \delta} \ge \frac{8}{1 + 4g} \ge 4\ .
\end{equation*}
Thus, we have shown that the payoff of $r$ is at most
$2 - 4 g \cdot \delta$.
Thus the row player's regret is at most $1 - 4g \cdot \delta$.
\qed
\end{itemize}
\end{proof}

\subsection{Soundness}
We now suppose that $\omega(\fcal) < 1 - \delta/2$, and we will prove that all
$(1 - 4g \cdot \delta)$-UNE  provide social welfare at most $2 - 2g \cdot
\delta$. 
Throughout this subsection, we will fix $(\xbf, \ybf)$ to be a $(1 - 4g \cdot
\delta)$-UNE of $\mathcal{G}$. We begin by making a simple observation about the
amount of probability that is placed on the subgame $(R, C)$.

\begin{lemma}
\label{lem:problp}
If $\sw(\xbf, \ybf) > 2 - 2g \cdot \delta$, then
\begin{itemize}
\item $\xbf$ places at least $(1 - g \cdot \delta)$ probability on rows in $(R, C)$,
and 
\item $\ybf$ places at least $(1 - g \cdot \delta)$ probability on columns in~$(R, C)$.
\end{itemize}
\end{lemma}

\begin{proof}
We will prove the lemma for $\xbf$; the proof for $\ybf$ is entirely
symmetric. For the sake of contradiction, suppose that $\xbf$ places strictly
less than $(1 - g \cdot \delta)$ probability on rows in $(R, C)$. Observe that every subgame
of $\mathcal{G}$ other than $(R, C)$ is a zero-sum game. Thus, any probability
assigned to these sub-games contributes nothing to the social welfare. On the
other hand, the payoffs in $(R, C)$ are at most $1$. So, even if the column
player places all probability on columns in $C$, the social welfare $\sw(\xbf,\ybf)$ 
will be
strictly less than $2 \cdot (1 - g \cdot \delta) + g \cdot \delta \cdot 0 = 2 -
2 g \cdot \delta$, a contradiction. \qed
\end{proof}

So, for the rest of this subsection, we can assume that both $\xbf$ and~$\ybf$
place at least $1 - g \cdot \delta$ probability on the subgame $(R, C)$. We will
ultimately show that, if this is the case, then both players have payoff at most
$1 - \frac{1}{2} \cdot \delta + m g \cdot \delta$ for some constant $m$ that
will be derived during the proof. Choosing $g = 1/(2m + 2)$ then ensures that
both players have payoff at most $1 - g \cdot \delta$, and therefore that the
social welfare is at most $2 - 2g \cdot \delta$.

\paragraph{\bf A two-prover game}
We use $(\xbf, \ybf)$ to create a two-prover game. First, we define two
distributions that capture the marginal probability that a question is played by
$\xbf$ or $\ybf$. Formally, we define a distribution $\xbf'$ over $X$ and a 
distribution $\ybf'$ over $Y$ such that for all $x \in X$ and $y \in Y$ we have
$\xbf'(x) = \sum_{a \in A} \xbf(x, a),$ and 
$\ybf'(y) = \sum_{b \in B} \ybf(y, b).$
By Lemma~\ref{lem:problp}, we can assume that $\| \xbf' \|_1 \ge 1 - g \cdot \delta$
and $\| \ybf' \|_1 \ge 1 - g \cdot \delta$. 

Our two-prover game will have the same question sets, answer sets, and
verification function as \fcal, but a different distribution over the question
sets. Let $\txy = (X, Y, A, B, \mathcal{D}, V)$, where $\mathcal{D}$ is the
product of $\xbf'$ and $\ybf'$. Note that we have cheated slightly here, since
$\mathcal{D}$ is not actually a probability distribution. If $\| \mathcal{D}
\|_1 = c < 1$, then we can think of this as Arthur having a $1 - c$ probability
of not sending any questions to the Merlins and awarding them payoff $0$.

The strategies $\xbf$ and $\ybf$ can also be used to give a us a strategy for
the Merlins in $\txy$. Without loss of generality, we can assume that for each
question $x \in X$ there is exactly one answer $a \in A$ such that $\xbf(x, a) >
0$, because if there are two answers $a_1$ and $a_2$ such that $\xbf(x,
a_1) > 0$ and $\xbf(x, a_2) > 0$, then we can shift all probability onto the
answer with (weakly) higher payoff, and (weakly) improve the payoff to the
row player. Since
$(R, C)$ is cooperative, this can only improve the payoff of the columns in $(R,
C)$, and since the row player does not move probability between questions, the
payoff of the columns in $(-D_2, D_2)$ does not change either. Thus, after
shifting, we arrive at a $(1 - 4g \cdot \delta)$-UNE of $\mathcal{G}$
whose social welfare is at least as good as $\sw(\xbf, \ybf)$. Similarly,
we can assume that for each question $y \in Y$ there is exactly one answer $b
\in B$ such that $\ybf(y, b) > 0$. 

So, we can define a strategy $s_\xbf$ for \mone in the following way. For each
question $x \in X$, the strategy $s_\xbf$ selects the unique answer $a \in A$
such that~$\xbf(x, a) > 0$. The strategy $s_\ybf$ for \mtwo is defined
symmetrically.  

We will use $\txy$ as an intermediary between $\mathcal{G}$ and $\fcal$ by
showing
that the payoff of $(\xbf, \ybf)$ in $\mathcal{G}$ is close to the
payoff of $(s_\xbf, s_\ybf)$ in $\txy$, and that the payoff of
$(s_\xbf, s_\ybf)$ in $\txy$ is close to the payoff of $(s_\xbf, s_\ybf)$ in
\fcal. Since we have a bound on the payoff of any pair of strategies in \fcal,
this will ultimately allow us to bound the payoff to both players when $(\xbf,
\ybf)$ is played in $\mathcal{G}$.

\paragraph{\bf Relating $\mathcal{G}$ to \txy}
For notational convenience, let us define $p_r(\mathcal{G}, \xbf, \ybf)$ and
$p_c(\mathcal{G}, \xbf, \ybf)$ to be the payoff to the row player and column
player, respectively, when $(\xbf, \ybf)$ is played in $\mathcal{G}$. We begin
by showing that the difference between $p_r(\mathcal{G}, \xbf, \ybf)$ and
$p(\txy, s_\xbf, s_\ybf)$ is small. Once again we prove this for the payoff of
the row player, but the analogous result also holds for the 
column player.

\begin{lemma}
\label{lem:pfcomp}
We have 
$| p_r(\mathcal{G}, \xbf, \ybf) - p(\txy, s_\xbf, s_\ybf) | \le 4 g \cdot
\delta.$
\end{lemma}
\begin{proof}
By construction, $p(\txy, s_\xbf, s_\ybf)$ is equal to the payoff that the row
player obtains from the subgame $(R, C)$, and so we have $p(\txy, s_\xbf,
s_\ybf) \le p_r(\mathcal{G}, \xbf, \ybf) $. On the other hand, since the row
player places at most $g \cdot \delta$ probability on rows not in $(R, C)$, and
since these rows have payoff at most $\frac{4}{1 + 4g \cdot \delta} < 4$, we
have $ p_r(\mathcal{G}, \xbf, \ybf) \le p(\txy, s_\xbf, s_\ybf) + 4g \cdot
\delta$. \qed
\end{proof}

\paragraph{\bf Relating $\txy$ to $\fcal$}

First we show that if $(\xbf, \ybf)$ is indeed a $(1 - 4g \cdot
\delta)$-UNE, then  $\xbf'$ and $\ybf'$ must be close to uniform over the
questions. We prove this for~$\ybf'$, but the proof can equally well be
applied to $\xbf'$. The idea is that, if $\ybf'$ is sufficiently far
from uniform, then there is  set $B \subseteq Y$ of $|Y|/2$ columns where
$\ybf'$ places significantly more than $0.5$ probability. This, in turn, means
that the row of $(D_1, -D_1)$ that corresponds to $B$, will have payoff at least
$2$, while the payoff of $(\xbf, \ybf)$ can be at most $1 + 3g \cdot \delta$,
and so $(\xbf, \ybf)$ would not be a $(1 - 4g \cdot \delta)$-UNE. We formalise
this idea in the following lemma.
Define $\ubf_X$ to be the uniform distribution over $X$, and $\ubf_Y$ to be the
uniform distribution over $Y$.

\begin{lemma}
\label{lem:dbound}
We have
$\| \ubf_Y - \ybf' \|_1 < 16 g \cdot \delta$ 
and $\| \ubf_X - \xbf' \|_1 < 16 g \cdot \delta$.
\end{lemma}

We begin by proving an auxiliary lemma.

\begin{lemma}
\label{lem:subset}
If $\| \ubf_Y - \ybf' \|_1 \ge c$
then there exists a set $B \subseteq Y$ of size $|Y|/2$ such that 
\begin{equation*}
\sum_{i \in B} \ybf'_i > \frac{1}{2} + \frac{c}{4} - 2 g \cdot \delta.
\end{equation*}
\end{lemma}
\begin{proof}
We first define $\dbf = \ybf' - \ubf_Y$, and then we partition $Y$ as follows
\begin{align*}
U &= \{ y \in Y \; : \; \dbf_y > \frac{1}{|Y|} \}, \\
L &= \{ y \in Y \; : \; \dbf_y \le \frac{1}{|Y|} \}.
\end{align*}
Since $\| \ybf' \|_1 \ge 1 - g \cdot \delta$ and $\| \ubf \|_1 = 1$, we have
that 
\begin{align*}
\sum_{y \in U} \dbf_y &\ge c/2 - g \cdot \delta, \\
\sum_{y \in L} \dbf_y &\le -c/2 + g \cdot \delta.
\end{align*}
We will prove that there exists a set $B \subseteq Y$  of size $|Y|/2$
such that $\sum_{y \in B} \dbf_y \ge c/4 - g \cdot \delta$. 

We have two cases to consider, depending on the size of $U$.
\begin{itemize}
\item First suppose that $|U| > |Y|/2$. If this is the case, then
there must exist a set $B \subseteq U$ with $|B| = |U|/2$ and
$\sum_{i \in B} \dbf_i \ge c/4 - g \cdot \delta$. We can then add arbitrary
columns from $U \setminus B$ to $B$ in order to make $|B| = |Y|/2$, and
since $\dbf_i > 0$ for all $i \in U$, this cannot decrease $\sum_{i \in
B} \dbf_i$. Thus, we have completed the proof for this case.

\item Now suppose that $|U| \le |Y|/2$. If this is the case, then there must
exist a set $C \subseteq L$  with $|C| = |L|/2$ and $\sum_{i \in C} \dbf_i \ge
-\frac{c}{4} + g \cdot \delta$. 
So, let $C' \subseteq C$ be an arbitrarily chosen subset such
that $|C'| + |U| = |Y|/2$. This is possible since $|L| = |Y| - |U|$  and hence 
$|L|/2 = |Y|/2 - |U|/2$, which implies that $|L|/2 + |U| > |Y|/2$. Setting $B = C'
\cup U$ therefore gives us a set with $|B| = |Y|/2$ such that
\begin{align*}
\sum_{i \in B} \dbf_i &\ge \left(c/2 - g \cdot \delta\right) - \left(c/4 + g \cdot \delta
\right) \\
&= c/4 - 2 \cdot g \cdot \delta\ .
\end{align*}
So we have completed the proof of this case, and the lemma as a whole. \qed
\end{itemize}

\end{proof}

We can now proceed with the proof of Lemma~\ref{lem:dbound}.
\begin{proof}[Proof of Lemma~\ref{lem:dbound}]
Suppose, for the sake of contradiction that one of these two properties
fails. Without loss of generality, let us assume that 
$\| \ubf_Y - \ybf' \|_1 \ge c$. We will show that the row player can gain more
than $1$ in payoff by deviating to a new strategy, which will show that
$(\xbf, \ybf)$ is not a~$1$-UNE, contradicting our assumption that it is 
a $(1 - 4 g \cdot \delta)$-UNE.

By assumption, $\xbf$ places at least $1 - g \cdot \delta$ probability on rows
in $(R, C)$. The maximum payoff in $R$ is $1$, and the maximum payoff in $-D_2$
is $0$. On the one hand, the rows in $D_2$ give payoff at most $8/(2 + g \cdot
\delta) \le 4$. 
So the row player's payoff under $(\xbf, \ybf)$ is bounded by 
\begin{equation*}
\left(1 - g \cdot \delta\right) \cdot 1 + \left(g \cdot \delta\right) \cdot  4 = 1 + 3 g \cdot
\delta.
\end{equation*}

On the other hand, we can apply Lemma~\ref{lem:subset} with $c = 16 g \cdot
\delta$ to find a set
$B \subseteq Y$ such that 
\begin{align*}
\sum_{i \in B} \ybf'(i) &> \frac{1}{2} + \frac{16g \cdot \delta}{4} - 2 g \cdot
\delta. \\
&= \frac{1}{2} + 2 g \cdot \delta \\
&= \frac{1 + 4g \cdot \delta}{2}.
\end{align*}
So, let $r_B$ be the row of $D_1$ that corresponds to $B$. This row has payoff
$\frac{8}{2 + g \cdot \delta}$ for every entry in $B$. So, the payoff of row
$r_B$ must be at least
\begin{equation*}
\left( \frac{1 + 4g \cdot \delta}{2} \right) \cdot \left(
\frac{4}{1 + 4g \cdot \delta} \right) = 2.
\end{equation*}
Thus, the row player can deviate to $r_B$ and increase his payoff by at least $1
- 3 g \cdot \delta$, and $(\xbf, \ybf)$ is not a $(1 - 4 g \cdot \delta)$-UNE.
  \qed
\end{proof}

With Lemma~\ref{lem:dbound} at hand, we can now prove that the difference
between $p(\txy, s_\xbf, s_\ybf)$ and $p(\fcal, s_\xbf, s_\ybf)$ must be small.
This is because the question distribution $\mathcal{D}$ used in \txy is a
product of two distributions that are close to uniform, while the question
distribution $\mathcal{U}$ used in \fcal is a product of two uniform
distributions. In the following lemma, we show that if we transform
$\mathcal{D}$ into $\mathcal{U}$, then we do not change the payoff of $(s_\xbf,
s_\ybf)$ very much.

\begin{lemma}
\label{lem:tfcomp}
We have
$| p(\txy, s_\xbf, s_\ybf) -  p(\fcal, s_\xbf, s_\ybf) | \le 64 g \cdot \delta.$
\end{lemma}

\begin{proof}
The distribution used in \fcal is the product of $\ubf_Y$ and $\ubf_X$, while the
distribution used in \txy is the product of $\ybf$' and $\xbf'$. Furthermore,
Lemma~\ref{lem:dbound} tells us that $\| \ubf_Y - \ybf' \|_1 < 16 g \cdot
\delta$ and $\| \ubf_X - \xbf' \|_1 < 16 g \cdot \delta$. Our approach is to
transform $\ubf_X$ to $\xbf'$ while bounding the amount that $p(\fcal, s_\xbf,
s_\ybf)$
changes. Once we have this, we can apply the same transformation to $\ubf_Y$ and
$\ybf'$.

Consider the effect of shifting probability from a question $x_1 \in X$ to a
different question $x_2 \in X$. Since all entries of $V$ are in $\{0, 1\}$, if
we shift~$q$ probability from $x_1$ to $x_2$, then $p(\fcal, s_\xbf, s_\ybf)$ can
change by at most $2q$. This bound also holds if we remove probability from
$x_1$ without adding it to $x_2$ (which we might do since $\| \xbf \|_1$ may not be
$1$.) Thus, if we shift probability to transform $\ubf_X$ into $\xbf'$, then we
can change $p(\fcal, s_\xbf, s_\ybf)$ by at most $ 32 g \cdot \delta$.

The same reasoning holds for transforming $\ubf_Y$ into $\ybf'$. This means that
we can transform \fcal to \txy while changing the payoff of $(s_\xbf, s_\ybf)$ by at
most $64 g \cdot \delta$, which completes the proof. \qed
\end{proof}

\paragraph{\bf Completing the soundness proof} 

The following lemma uses the bounds derived in Lemmas~\ref{lem:pfcomp}
and~\ref{lem:tfcomp}, along with a suitable setting for~$g$, to bound the payoff
of both players when $(\xbf, \ybf)$ is played in $\mathcal{G}$.

\begin{lemma}
\label{lem:sound}
If $g = \frac{1}{138}$, then both players have payoff at most $1 - g \cdot
\delta$ when $(\xbf, \ybf)$ is played in $\mathcal{G}$.
\end{lemma}
\begin{proof}
Lemmas~\ref{lem:pfcomp} and~\ref{lem:tfcomp} tell us that
\begin{align*}
| \; p_r(\mathcal{G}, \xbf, \ybf) - p(\txy, s_\xbf, s_\ybf) | &\le 4g \cdot \delta, \\
| \; p(\txy, s_\xbf, s_\ybf) - p(\fcal, s_\xbf, s_\ybf) | &\le 64g \cdot \delta. 
\end{align*}
Hence, we have
$| \; p_r(\mathcal{G}, \xbf, \ybf) - p(\fcal, s_\xbf, s_\ybf) | \le 68g \cdot
\delta$. However, we know that $p(\fcal, s_\xbf, s_\ybf) \le 1 - \delta/2$. So,
if we set $g = \frac{1}{138}$, then we we will have that
\begin{align*}
p_r(\mathcal{G}, \xbf, \ybf) &\le 1 - \frac{1}{2} \cdot \delta + \frac{68}{138} \cdot
\delta \\
&= 1 - \frac{1}{138} \cdot \delta\\
&= 1 - g \cdot \delta.
\end{align*}
\qed
\end{proof}

\noindent Hence, we have proved that $\sw(\xbf, \ybf) \le 2 - 2g \cdot \delta$.

\subsection{The result}

We can now state the theorem that we have proved in this section.
We first rescale the game so that it lies in $[0, 1]$.
The maximum payoff in $\mathcal{G}$ is $\frac{4}{1 + 4g \cdot
\delta} \le 4$, and the minimum payoff is $- \frac{4}{1 + 4g \cdot \delta}
\ge -4$. To rescale this game, we add $4$ to all the payoffs, and then divide by
$8$. Let us refer to the scaled game as $\mathcal{G}_s$. Observe that an
$\eps$-UNE in $\mathcal{G}$ is a $\frac{\eps}{8}$-NE in $\mathcal{G}_s$ since
adding a constant to all payoffs does not change the approximation guarantee, but
dividing all payoffs by a constant does change the approximation guarantee. So,
we have the following theorem.

\begin{theorem}
\label{thm:main}
If ETH holds, then there exists a constant $\delta$ below which
the problem $(\frac{1 - 4 g \cdot \delta}{8})$-NE
$(\frac{g}{4}\cdot \delta)$-SW, where $g = \frac{1}{138}$, requires 
$n^{\widetilde \Omega(\log n)}$ time.
\end{theorem}

\begin{proof}
By Lemmas~\ref{lem:complete} and~\ref{lem:sound}, we have 
\begin{itemize}
\item if $\omega(\mathcal{F}) = 1$ then there exists a 
$(1 - 4g \cdot \delta)$-UNE of $\mathcal{G}$ with social welfare $1 + 1 = 2$.
In the rescaled game this translates to a
$(\frac{1 - 4 g \cdot \delta}{8})$-NE of $\mathcal{G}_s$ with social welfare
$\frac{1+4}{8} + \frac{1+4}{8} = \frac{10}{8}$.
\item if $\omega(\mathcal{F}) < 1 - \delta$ then all
$(1 - 4 g \cdot \delta)$-UNE of $\mathcal{G}$ have social welfare at most $(1 - g
\cdot \delta) + (1 - g \cdot \delta) = 2 - 2g \cdot \delta$.
After rescaling, we have that all
$(\frac{1 - 4 g \cdot \delta}{8})$-NE of $\mathcal{G}_s$ have social welfare
social welfare at most 
\begin{equation*}
\frac{5 - g \cdot \delta}{8} + \frac{5 - g \cdot \delta}{8} = 
\frac{10}{8} - \frac{g \cdot \delta}{4}
\end{equation*}
\end{itemize}
By Theorem~\ref{thm:qplower}, assuming ETH we require 
$|\mathcal{F}|^{\widetilde{\Omega}(\log 
 |\mathcal{F}|)}$ time to decide whether the value of $\mathcal{F}$ is $1$ or $1
- \delta$ for some small constant $\delta$. Thus, we also require 
$n^{\widetilde{\Omega}(\log 
 |n|)}$ to solve the problem 
$(\frac{1 - 4 g \cdot \delta}{8})$-NE $(\frac{g}{4}\cdot \delta)$-SW. \qed

\end{proof}

\section{Hardness results for other decision problems}
\label{sec:problems}

In this section we study a range of decision problems associated with
approximate equilibria. Table~\ref{tab:problems} shows all of the decision
problems that we consider. Most are known to be \NP-complete for the case of
exact Nash equilibria~\cite{GZ89,CS08}. For each problem in
Table~\ref{tab:problems}, the input includes a bimatrix game and a quality of
approximation $\eps \in (0,1)$.  We consider decision problems related to both
\eps-NE and \eps-WSNE.  Since \eps-NE is a weaker solution concept than \eps-WSNE, 
i.e., every \eps-WSNE is an \eps-NE, the
hardness results for \eps-NE imply the same hardness for \eps-WSNE.  We consider
problems for \eps-WNSE only where the corresponding problem for \eps-NE is
trivial. For example,  observe that deciding if there is an \eps-NE with large
support is a trivial problem, since we can always add a tiny amount of
probability to each pure strategy without changing our expected payoff very
much.

Our conditional quasi-polynomial lower bounds will hold for all $\eps<\frac{1}{8}$, 
so let us fix $\eps^* < \frac{1}{8}$ for the rest of this section.
Using Theorem~\ref{thm:main}, we compute from~$\eps^*$ the parameters
$n$ and $\delta$ that we require to apply Theorem~\ref{thm:main}.
In particular, set $\delta^*$ to solve 
$\eps^* = (\frac{1 - 4 g \cdot \delta^*}{8})$,
and choose $n^*$ as $\frac{1}{\delta^*}$.
Then, for $n>n^*$ and $\delta=\delta^*$ we can apply Theorem~\ref{thm:main} to
bound the social welfare achievable if $\omega(\fcal) < 1 - \delta^*$ as
$$\ufrak := \frac{10}{8} - \frac{1}{522}\delta^*.$$
Theorem~\ref{thm:main} implies that in order to decide whether the game
$\gcal_s$ possesses an $\eps^*$-NE that yields social welfare strictly greater
than $\ufrak$ requires $n^{\tilde{O}(\log n)}$ time, where $\delta$ no longer
appears in the exponent since we have fixed it as the constant~$\delta^*$.

Problem~\ref{probc:largep} asks to decide whether a bimatrix game possesses
an $\eps^*$-NE  where the expected payoff for each player is at least $u$, 
where $u$ is an input to the problem. When we 
set $u = \frac{5}{8}$,
the conditional hardness of this problem is an immediate corollary of Theorem~\ref{thm:main}.
 

For Problems~\mbox{\ref{probc:restrictedweak}~-~\ref{probc:maxsupport}}, we use
$\gcal_s$ to construct a new game \gcalp, which adds one row \ifrak and one
column \jfrak to $\gcal_s$. The payoffs  are defined using the constants $\ufrak$
and~$\eps^*$, as shown in Figure~\ref{fig:gcalp}.

\begin{figure}[tbph!]
\begin{center}
\setlength{\extrarowheight}{2pt}
\hspace{-10mm}
\resizebox{0.7\textwidth}{!}{ 
\begin{tabular}{cc|P{1.5cm}|c|P{1.5cm}|P{1.75cm}|}
	\multirow{6}{*}{\Large \gcalp =}
	& \multicolumn{1}{c}{} & 
	  \multicolumn{1}{c}{} & 
	  \multicolumn{1}{c}{} & 
	  \multicolumn{1}{c}{} & 
	  \multicolumn{1}{c}{\large $\jfrak$} \\\cline{3-6}
	&          & \multicolumn{3}{|c|}{}      & $0, \frac{5}{8}+\eps^*$ \\ \cline{6-6}
	&          & \multicolumn{3}{|c|}{\Large $\gcal_s$} & $\vdots$ \\ \cline{6-6}
	&          & \multicolumn{3}{|c|}{}      & $0, \frac{5}{8}+\eps^*$ \\ \cline{3-6}
	& {\large $\ifrak$} & $\frac{5}{8}+\eps^*, 0$ & $\cdots$ & $\frac{5}{8}+\eps^*, 0$ & $1, 1$ \\\cline{3-6}
\end{tabular}
}
\end{center}
\medskip
\caption{The game \gcalp.}
\label{fig:gcalp}
\end{figure}


In \gcalp, the expected payoff for the row player for~\ifrak
is at least $\frac{5}{8}+\eps^*$ irrespective of the column player's strategy.
Similarly, 
the expected payoff for~\jfrak is at 
least~${\frac{5}{8}+\eps^*}$ irrespective of the row player's strategy.
This means that:
\begin{itemize}
\item
If $\gcal_s$ possesses an $\eps^*$-NE with social welfare $\frac{10}{8}$,
then \gcalp possesses at least one $\eps^*$-NE where 
the players do not play the pure strategies \ifrak and~\jfrak.
\item
If every $\eps^*$-NE of $\gcal_s$ yields social welfare at most 
\ufrak, then in \emph{every} $\eps^*$-NE of \gcalp, the players place almost 
all of their probability on \ifrak and \jfrak respectively. 
Note that $(\ifrak,\jfrak)$ is a pure exact Nash equilibrium.
\end{itemize}
 
Problem~\ref{probc:restrictedweak} asks whether a bimatrix game
possesses an $\eps$-NE where the row player plays with positive probability only
strategies in a given set $S$. 
Let $S_R$  ($S_C$) denote the set of pure strategies
available to the row (column) player from the subgame~$(R, C)$ of $\gcal_s$.
To show the hardness of Problem~\ref{probc:restrictedweak}, we will set we set $S=S_R$.

Recall that $\gcal_s$ is created from \fcal.
First, we prove in Lemma~\ref{lem:groupa-yes} that if~${\omega(\fcal)=1}$, then
$\gcal'$ possesses an $\eps^*$-NE such that the answer to
Problem~\ref{probc:restrictedweak} is ``Yes''.  
Note that we actually argue in Lemma~\ref{lem:groupa-yes} about the existence of
an $\eps^*$-WSNE, since this stronger claim will be useful when we come to deal with
Problems~\ref{probc:totalmaxsupport}~-~\ref{probc:maxsupport}.
\begin{lemma}
\label{lem:groupa-yes}
If $\omega(\fcal)=1$, then $\gcal'$ possesses an $\eps^*$-WSNE $(\xbf, \ybf)$
such that ${\supp(\xbf) \subseteq S_R}$.
Under $(\xbf, \ybf)$, both players get payoff $\frac{5}{8}$, so $\sw(\xbf, \ybf) = \frac{10}{8}$.
Moreover, $|\supp(\xbf)| = |X|$ and $\max_i \xbf_i \leq \frac{1}{|X|}$, where
$X$ is the question set of \mone in \fcal.
\end{lemma}
\begin{proof}
The proof of Lemma~\ref{lem:complete} shows that, if $\omega(\fcal)=1$, then
$\gcal$ possesses an $\eps^*$-WSNE $(\xbf, \ybf)$ where the expected payoff
for each player is $1$ and ${\supp(\xbf) \subseteq S_R}$.
The reason that $(\xbf, \ybf)$ is well supported is that all rows in $\supp(\xbf)$
have equal expected payoff.
Moreover, \xbf is a uniform mixture over a pure strategy set of size $|X|$, where
$X$ is the question set of \mone in \fcal. 
Since $\gcal_s$ is obtained from~$\gcal$ by adding 4 to the payoffs and dividing
by 8, $(\xbf,
\ybf)$ is as an $\eps^*$-NE in $\gcal_s$ where each player has payoff $\frac{5}{8}$.
To complete the proof we show that $(\xbf,\ybf)$ in an $\eps^*$-NE for
$\gcal'$, which is the same as $\gcal_s$ apart from the additional 
pure strategies~$\ifrak$ and~$\jfrak$. Since $\ifrak$ and~$\jfrak$ yield
payoff $\frac{5}{8}+\eps^*$, but not more, the claim holds.
\qed
\end{proof}

Next we prove that if  $\omega(\fcal) < 1-\delta^*$, then the
answer to Problem~\ref{probc:restrictedweak} is ``No''.

\begin{lemma}
\label{lem:groupa-no}
If $\omega(\fcal) < 1-\delta^*$, then in every $\eps^*$-NE $(\xbf, \ybf)$
of \gcalp it holds that $\xbfi > 1 - \frac{\eps^*}{1-\eps^*}$ and 
$\ybfj > 1 - \frac{\eps^*}{1-\eps^*}$ .
\end{lemma}
\begin{proof}


Let $\gcal_s := (P, Q)$ and
suppose that $(\xbf, \ybf)$ is an
$\eps^*$-NE of \gcalp. From Theorem~\ref{thm:main} we know that 
if $\omega(\fcal) < 1-\delta^*$, then in any 
$\eps^*$-NE of $\gcal_s$ we have that each player gets payoff at most 
$\frac{\ufrak}{2} < \frac{5}{8}$. 
Under $(\xbf, \ybf)$ in \gcalp the row player gets payoff 
\begin{align*}
\xbf^T P \ybf & < (1-\xbf_\ifrak)\cdot(1-\ybf_\jfrak)\cdot \frac{5}{8} + 
\xbf_\ifrak \cdot(1-\ybf_\jfrak) (\frac{5}{8}+ \eps^*) + \xbf_\ifrak\cdot\ybf_\jfrak \\
& = \xbf_\ifrak\cdot \left( (1-\ybf_\jfrak) \cdot \eps^* + \ybf_\jfrak \right)
+ (1-\ybf_\jfrak) \cdot \frac{5}{8}.
\end{align*}
From the pure strategy \ifrak, the row player gets
$$P_\ifrak\cdot\ybf = (1-\ybf_\jfrak)(\frac{5}{8}+\eps^*) + \ybf_\jfrak.$$ 
In 
order for $(\xbf, \ybf)$ to be an $\eps^*$-NE it must hold that 
$\xbf^T P \ybf \geq P_\ifrak\ybf - \eps^*$. Using the upper bound
on $\xbf^T P\ybf$ that we just derived, we get:
\begin{align}
\label{eq:xbound}
\xbf_\ifrak > 1 - \frac{\eps^*}{(1-\ybf_\jfrak)\cdot \eps^* + \ybf_\jfrak}.
\end{align}
By symmetry, we also have that the column player must play
\jfrak with probability:
\begin{align}
\label{eq:ybound}
\ybf_\jfrak > 1 - \frac{\eps^*}{(1-\xbf_\ifrak)\cdot \eps^* + \xbf_\ifrak}.
\end{align}

Recall that in this section $\eps^*$ is a constant.
Observe that
the right-hand side of~\eqref{eq:ybound} is increasing in
$\xbf_\ifrak$, and we can thus use it to replace $\xbf_\ifrak$ in~\eqref{eq:ybound}
as follows:
\begin{align*}
\ybf_\jfrak & > 1- \frac{\eps^*}{(1-1+\frac{\eps^*}{(1-\ybf_\jfrak)\eps^*+\ybf_\jfrak})\eps^*+
1-\frac{\eps^*}{(1-\ybf_\jfrak)\eps^*+\ybf_\jfrak}}\\
 & = 1- \frac{\eps^*}{\frac{\eps^{*^2}}{(1-\ybf_\jfrak)\eps^*+\ybf_\jfrak}+
1-\frac{\eps^*}{(1-\ybf_\jfrak)\eps^*+\ybf_\jfrak}}\\
 & = 1- \frac{(1-\ybf_\jfrak)\eps^{*^2} + \ybf_\jfrak \eps^*}{\eps^{*^2}+(1-\ybf_\jfrak)\eps^*+\ybf_\jfrak-\eps^*}.
\end{align*}
Noting that $(\eps^{*^2}+(1-\ybf_\jfrak)\eps^*+\ybf_\jfrak-\eps^*) \ge 0$, by rearranging we get that
$$\ybf_\jfrak^2(1-\eps^*) + \ybf_\jfrak(2\eps^*-1) > 0.$$
Then, since $\eps^*<\frac{1}{8}$, we have $1-\eps ^*> 0$, and we get that 
\begin{align*}
\ybf_\jfrak > \frac{1-2\eps^*}{1-\eps^*} = 1 - \frac{\eps^*}{1-\eps^*}.
\end{align*}
By symmetry, we have $\xbf_\ifrak > 1 - \frac{\eps^*}{1-\eps^*}$, which completes the proof.\qed
\end{proof}

Next we recall Problems~\ref{probc:difcomp} and~\ref{probc:maxprob} and we show
that, as for Problem~\ref{probc:restrictedweak}, Lemmas~\ref{lem:groupa-yes}
and~\ref{lem:groupa-no} can also be used to immediately show that there are
instances of these decision problems where the answer is ``Yes'' if and only if
$\omega(\fcal)=1$.

Given two probability distributions $\xbf$ and $\xbf'$, the \emph{Total
Variation (TV)} distance between them is $\max_i \{ |\xbf_i - \xbf'_i| \}$. We
define the TV distance between two strategy profiles $(\xbf, \ybf)$ and $(\xbf',
\ybf')$ to be the maximum over the TV distance of $\xbf$ and $\xbf'$ and the TV
distance of $\ybf$ and $\ybf'$.  Problem~\ref{probc:difcomp} asks whether a
bimatrix game possesses two \eps -NEs with TV distance at least $d$.  
In order to apply Lemmas~\ref{lem:groupa-yes}
and~\ref{lem:groupa-no}, we will set $d=1 - \frac{\eps^*}{1-\eps^*}$.
Then an instance $\gcalp$ of Problem~\ref{probc:difcomp} is ``Yes'' when
$\omega(\fcal)=1$ since the $\eps^*$-NE $(\xbf, \ybf)$ identified in
Lemma~\ref{lem:groupa-yes}, has TV distance one from the pure exact Nash
equilbrium $(\ifrak,\jfrak)$.
Lemma~\ref{lem:groupa-no} says that, if $\omega(\fcal) < 1-\delta^*$, 
every $\eps^*$-NE $(\xbf, \ybf)$ of \gcalp has $\xbfi > 1 - \frac{\eps^*}{1-\eps^*}$ and so 
all $\eps^*$-NE are within TV distance $1-\frac{\eps^*}{1-\eps^*}$ of each other.

Problem~\ref{probc:maxprob} asks to decide whether there exists an \eps-NE
where the row player does not play any pure strategy with probability more than
$p$. 
For this problem, we set $p = \frac{1}{|X|}$, where $X$ is the question set for 
\mone. According Lemma~\ref{lem:groupa-yes}, if $\omega(\fcal)=1$, then an instance
$\gcalp$ of Problem~\ref{probc:maxprob} is a ``Yes''. 
Lemma~\ref{lem:groupa-no} says that, if $\omega(\fcal) < 1-\delta^*$, 
then for all $\eps^*$-NE $(\xbf, \ybf)$ of \gcalp, $\max_i\xbf_i \geq
\xbf_\ifrak > 1 - \frac{\eps^*}{1-\eps^*} > \frac{1}{|X|}$.

Problem~\ref{probc:smalltp}
asks whether a bimatrix game possesses an $\eps$-NE with social welfare at most
$v$, and Problem~\ref{probc:smallp} asks whether a bimatrix game possesses an
\eps -NE where the expected payoff of the row player is at most $u$.
We fix $v = \frac{10}{8}$ for Problem~\ref{probc:smalltp},
and for Problem~\ref{probc:smallp} we fix $u = \frac{5}{8}$.
As we have already explained in the proof of Lemma~\ref{lem:groupa-yes}, 
if $\omega(\fcal)=1$, then there is an $\eps^*$-NE for $\gcal'$ such that the expected
payoff for each player is $\frac{5}{8}$ and thus the social welfare is $\frac{10}{8}$.
So, if $\omega(\fcal)=1$, then the answer to Problems~\ref{probc:smalltp} and~\ref{probc:smallp}
is ``Yes''. On the other hand, from the proof of Lemma~\ref{lem:groupa-no} we know that if 
$\omega(\fcal) < 1-\delta^*$, then  in 
any $\eps^*$-NE of \gcalp both players play the strategies \ifrak and \jfrak 
with probability at least $1 - \frac{\eps^*}{1-\eps^*}$. So, each player gets payoff at least 
$(1 - \frac{\eps^*}{1-\eps^*})^2 > \frac{5}{8}$, since $\eps^* < \frac{1}{8}$, from their 
pure strategies \ifrak and \jfrak.
So, if $\omega(\fcal)<1-\delta^*$, then the answer to Problems~\ref{probc:smalltp} 
and~\ref{probc:smallp} is ``No''. 

Problems~\ref{probc:totalmaxsupport}~-~\ref{probc:maxsupport} relate to 
deciding if there exist approximate \emph{well-supported} equilibria with large supports 
(for \eps-NE these problems would be trivial).
Problem~\ref{probc:totalmaxsupport} asks whether a bimatrix game possesses an
\eps-WSNE~$(\xbf, \ybf)$ with $|\supp(\xbf)| + |\supp(\ybf)| \geq 2k$.
Problem~\ref{probc:minmaxsupport} asks whether a bimatrix game possesses an
{}\eps-WSNE~$(\xbf, \ybf)$ with $\min\{|\supp(\xbf)|, |\supp(\ybf)|\} \geq k$.
Problem~\ref{probc:maxsupport} asks whether a bimatrix game possesses an
\eps-WSNE $(\xbf, \ybf)$ with $|\supp(\xbf)| \geq k$.  
Recall that $X$ and~$Y$ are the question sets of \mone and \mtwo respectively that 
were used to define \fcal and in turn $\gcal_s$. 
We will fix $k = |X| = |Y|$ for all three problems.

If $\omega(\fcal)=1$, then Lemma~\ref{lem:groupa-yes}
says that there exists an $\eps^*$-WSNE $(\xbf, \ybf)$ for~\gcalp such that
$|\supp(\xbf)|= |\supp(\ybf)|=k$ and thus the answer to 
Problems~\ref{probc:totalmaxsupport}~-~\ref{probc:maxsupport} 
is ``Yes''. On the other hand, if $\omega(\fcal)<1-\delta$, then we will
prove that there is a unique $\eps^*$-WSNE where the row player plays only the pure strategy \ifrak and
the column player plays the pure strategy \jfrak.
\begin{lemma}
\label{lem:groupc}
If $\omega(\fcal)<1-\delta^*$, then there is a unique $\eps^*$-WSNE $(\xbf, \ybf)$ in \gcalp 
such that $\xbf_\ifrak=1$ and $\ybf_\jfrak=1$.
\end{lemma}
\begin{proof}
We consider only the case that $\omega(\fcal)<1-\delta^*$.
Then Lemma~\ref{lem:groupa-no} says that in 
every $\eps^*$-NE of \gcalp the column player plays the 
pure strategy \jfrak with probability at least $1 - \frac{\eps^*}{1-\eps^*}$. 
Against \jfrak, the row player gets $0$ for all pure strategies $i \neq \ifrak$ and
$1$ for \ifrak.
Thus, in any $\eps^*$-NE of \gcalp, for every pure strategy $i \neq \ifrak$, the
row player gets at most $\frac{\eps^*}{1-\eps^*}$ from every pure strategy
$i$, and the row player gets at least $1 - \frac{\eps^*}{1-\eps^*}$ from
\ifrak.
So, in every $\eps^*$-WSNE the row player must play only the pure strategy \ifrak since from 
every other pure strategy the player suffers regret at least $1 - \frac{2\eps^*}{1-\eps^*}$, which
is strictly larger than $\eps^*$ for every $\eps^* < \frac{1}{8}$. In turn, against~\ifrak,
every pure strategy $j \neq \jfrak$ for the column player yields zero payoff 
while the strategy \jfrak yields payoff 1. So, the unique $\eps^*$-WSNE of \gcalp
is $\xbf_\ifrak=1$ and~$\ybf_\jfrak=1$.
\qed
\end{proof}
Hence, when $\omega(\fcal)<1-\delta^*$ the answer to 
Problems~\ref{probc:totalmaxsupport}~-~\ref{probc:maxsupport} is ``No''. 
Thus, we have shown the following: 
\begin{theorem}
\label{thm:groupc}
Assuming ETH, any algorithm that solves the Problems~\ref{probc:restrictedweak}
-~\ref{probc:maxsupport}
for any constant $\eps< \frac{1}{8}$ requires $n^{\tilde{\Omega}(\log n)}$ time.
\end{theorem}

Finally, for Problem~\ref{probc:restricting}, we define a new game \gcalpp by
extending \gcalp. We add the new pure strategies \ifrakp for the row player and
\jfrakp for the column player. 
The payoffs are shown in Figure~\ref{fig:gcalpp}.  
Recall that Problem~\ref{probc:restricting} asks whether a bimatrix game
possesses an \eps-WSNE such that every strategy from a given set $S$ is played
with positive probability.

In order to prove our result we fix $S=\ifrakp$.
First, we prove that if ${\omega(\fcal)=1}$ then the game \gcalpp possesses an
$\eps^*$-WSNE $(\xbf, \ybf)$ such that $\ifrakp \in \supp(\xbf)$.  Then we prove
that if $\omega(\fcal)<1-\delta$, then for any $\eps^*$-WSNE $(\xbf, \ybf)$ it
holds that~${\ifrakp \notin \supp(\xbf)}$.


\begin{figure}[tbph!]
\hspace{1cm}
\setlength{\extrarowheight}{2pt}
\resizebox{0.7\textwidth}{!}{ 
\begin{tabular}{cc|P{1.5cm}|c|P{1.5cm}|P{1.75cm}|}
\multirow{6}{*}{\Large \gcalpp =}
& \multicolumn{1}{c}{} & 
  \multicolumn{1}{c}{} & 
  \multicolumn{1}{c}{} & 
  \multicolumn{1}{c}{} & 
  \multicolumn{1}{c}{\large $\jfrakp$} \\\cline{3-6}
&          & \multicolumn{3}{|c|}{}      & $\frac{5}{8}, \frac{5}{8}$ \\ \cline{6-6}
&          & \multicolumn{3}{|c|}{\Large \gcalp} & $\vdots$ \\ \cline{6-6}
&          & \multicolumn{3}{|c|}{}      & $\frac{5}{8}, \frac{5}{8}$ \\ \cline{3-6}
& {\large $\ifrakp$} & $\frac{5}{8}, \frac{5}{8}$ & $\cdots$ & $\frac{5}{8}, \frac{5}{8}$ & $0, 0$ \\\cline{3-6}
\end{tabular}
}
\medskip
\caption{The game \gcalpp.}
\label{fig:gcalpp}
\end{figure}


%
\begin{lemma}
\label{lem:ipinsup}
If $\omega(\fcal)=1$, 
then \gcalpp possesses an $\eps^*$-WSNE $(\xbf, \ybf)$ such that
${\ifrakp \in \supp(\xbf)}$.
\end{lemma}
\begin{proof}
Lemma~\ref{lem:groupa-yes} says that if $\omega(\fcal)=1$, then \gcalp
	possesses an $\eps^*$-WSNE $(\xbf', \ybf')$ that gives payoff $\frac{5}{8}$
	for each player, and $\xbf'$ is uniform on a set of size~$|X|$.  
We construct the required $\eps^*$-WSNE of \gcalpp from $(\xbf', \ybf')$ as
	follows.  We add~\ifrakp to the support of $\xbf'$ so that \xbf is a uniform
	mixture over $\supp(\xbf')\cup\ifrakp$.  For the column player, we extend
	$\ybf'$ by adding zero probability for \jfrakp.

Against \ybf, pure strategies in $\supp(\xbf')$ give payoff $\frac{5}{8}$, pure
	strategy \ifrak in \gcalp yields payoff $\frac{5}{8}+\eps^*$, and \ifrakp
	gives payoff $\frac{5}{8}$.  Thus, since $(\xbf', \ybf')$ is an
	$\eps^*$-WSNE of \gcalp, \xbf has pure regret at most $\eps^*$ against \ybf,
	as required.  What remains is to show that the pure regret of \ybf is no more
	than $\eps^*$ against \xbf.
	Recall that, in \gcalp, against $\xbf'$, 
	the payoff of each pure strategy in $\supp(\ybf')$ is~$\frac{5}{8}$.
	Now consider~\gcalpp. Since, against $\ifrakp$, the column player gets $\frac{5}{8}$ for all 
	$j \in \supp(\ybf)$, the column
	player still gets $\frac{5}{8}$ against~$\xbf$ for all $j \in \supp(\ybf)$.
	Moreover, against $\xbf$, the payoff of $\jfrak'$ is $\frac{|X|}{|X|+1}
	\cdot \frac{5}{8} < \frac{5}{8}$.
	Thus, since $(\xbf', \ybf')$ is an $\eps^*$-WSNE of \gcalp, we have that
	$(\xbf, \ybf)$ is an $\eps^*$-WSNE of \gcalpp with $\ifrakp \in \supp(\xbf)$, which
	completes the proof.
\qed
\end{proof}

\begin{lemma}
\label{lem:ipninsup}
If $\omega(\fcal)<1-\delta^*$, then 
for any~$\eps^*$-WSNE $(\xbf, \ybf)$ of \gcalpp it holds that 
$\ifrakp \notin \supp(\xbf)$.
\end{lemma}
\begin{proof}
We prove that the unique $\eps^*$-WSNE of \gcalpp is the pure profile~$(\ifrak, \jfrak)$.
Using exactly the same arguments as in the proof of Lemma~\ref{lem:groupa-no} we can 
prove that if $\omega(\fcal)<1-\delta^*$, then in any $\eps^*$-NE of \gcalpp it holds 
that $\xbfi > 1- \frac{\eps^*}{1-\eps^*}$ and $\ybfj > 1- \frac{\eps^*}{1-\eps^*}$.
Then, using exactly the same arguments as in Lemma~\ref{lem:groupc} we can get 
that the pure strategy \jfrak for the column player yields payoff at least 
$1- \frac{\eps^*}{1-\eps^*}$ while any other pure strategy, including \jfrakp, 
yields payoff at most $\frac{\eps^*}{1-\eps^*}$. Hence, in any $\eps^*$-WSNE of
\gcalpp the column player must play only the pure strategy \jfrak. Then, in order to
be in an $\eps^*$-WSNE the row player must play the pure strategy~\ifrak.
Our claim follows.
\qed
\end{proof}
 
The combination of Lemmas~\ref{lem:ipinsup} and~\ref{lem:ipninsup} gives the following theorem.
\begin{theorem}
\label{thm:main2}
Assuming the ETH, any algorithm that solves the Problem~\ref{probc:restricting} 
for any constant $\eps < \frac{1}{8}$ 
requires $n^{\mathrm{\tilde{\Omega}}(\log n)}$ time.
\end{theorem}


\paragraph{\bf Acknowledgements} We would like to thank Aviad Rubinstein for
alerting us to the existence of Theorem~\ref{thm:qplower}.


\newpage

\section*{\refname}
\bibliography{references}

\begin{thebibliography}{10}
\expandafter\ifx\csname url\endcsname\relax
  \def\url#1{\texttt{#1}}\fi
\expandafter\ifx\csname urlprefix\endcsname\relax\def\urlprefix{URL }\fi
\expandafter\ifx\csname href\endcsname\relax
  \def\href#1#2{#2} \def\path#1{#1}\fi

\bibitem{DGP09}
C.~Daskalakis, P.~W. Goldberg, C.~H. Papadimitriou, The complexity of computing
  a {N}ash equilibrium, SIAM Journal on Computing 39~(1) (2009) 195--259.

\bibitem{GZ89}
I.~Gilboa, E.~Zemel, {N}ash and correlated equilibria: Some complexity
  considerations, Games and Economic Behavior 1~(1) (1989) 80 -- 93.

\bibitem{CS08}
V.~Conitzer, T.~Sandholm, New complexity results about {N}ash equilibria, Games
  and Economic Behavior 63~(2) (2008) 621 -- 641.

\bibitem{GMVY}
J.~Garg, R.~Mehta, V.~V. Vazirani, S.~Yazdanbod, {ETR}-completeness for
  decision versions of multi-player (symmetric) {N}ash equilibria, in: Proc.\
  of ICALP, 2015, pp. 554--566.

\bibitem{BM16}
V.~Bil{\`{o}}, M.~Mavronicolas, A catalog of {$\exists\reals$}-complete
  decision problems about {N}ash equilibria in multi-player games, in: Proc.\
  of {STACS}, 2016, pp. 17:1--17:13.

\bibitem{BM17}
V.~Bil{\`{o}}, M.~Mavronicolas, {$\exists\reals$}-complete decision problems
  about symmetric nash equilibria in symmetric multi-player games, in: Proc.\
  of {STACS}, 2017, pp. 13:1--13:14.

\bibitem{BBM10}
H.~Bosse, J.~Byrka, E.~Markakis, New algorithms for approximate {N}ash
  equilibria in bimatrix games, Theoretical Computer Science 411~(1) (2010)
  164--173.

\bibitem{CDFFJS}
A.~Czumaj, A.~Deligkas, M.~Fasoulakis, J.~Fearnley, M.~Jurdzi\'nski, R.~Savani,
  Distributed methods for computing approximate equilibria, in: Proc.\ of
  {WINE}, 2016, pp. 15--28.

\bibitem{DMP07}
C.~Daskalakis, A.~Mehta, C.~H. Papadimitriou, Progress in approximate {N}ash
  equilibria, in: Proc.\ of EC, 2007, pp. 355--358.

\bibitem{DMP09}
C.~Daskalakis, A.~Mehta, C.~H. Papadimitriou, A note on approximate {N}ash
  equilibria, Theoretical Computer Science 410~(17) (2009) 1581--1588.

\bibitem{FGSS12}
J.~Fearnley, P.~W. Goldberg, R.~Savani, T.~B. S{\o}rensen, Approximate
  well-supported {N}ash equilibria below two-thirds, in: Proc.\ of SAGT, 2012,
  pp. 108--119.

\bibitem{KS10}
S.~C. Kontogiannis, P.~G. Spirakis, Well supported approximate equilibria in
  bimatrix games, Algorithmica 57~(4) (2010) 653--667.

\bibitem{TS08}
H.~Tsaknakis, P.~G. Spirakis, An optimization approach for approximate {N}ash
  equilibria, Internet Mathematics 5~(4) (2008) 365--382.

\bibitem{LMM03}
R.~J. Lipton, E.~Markakis, A.~Mehta, Playing large games using simple
  strategies, in: Proc.\ of EC, 2003, pp. 36--41.

\bibitem{R16}
A.~Rubinstein, Settling the complexity of computing approximate two-player
  {N}ash equilibria, in: Proc.\ of FOCS, 2016, pp. 258--265.

\bibitem{HK11}
E.~Hazan, R.~Krauthgamer, How hard is it to approximate the best {N}ash
  equilibrium?, {SIAM} J. Comput. 40~(1) (2011) 79--91.

\bibitem{ABC13}
P.~Austrin, M.~Braverman, E.~Chlamtac, Inapproximability of {NP}-complete
  variants of {N}ash equilibrium, Theory of Computing 9 (2013) 117--142.

\bibitem{BKW15}
M.~Braverman, Y.~Kun{-}Ko, O.~Weinstein, Approximating the best {N}ash
  equilibrium in \emph{n\({}^{\mbox{o}}\)}\({}^{\mbox{(log \emph{n})}}\)-time
  breaks the exponential time hypothesis, in: Proc.\ of {SODA}, 2015, pp.
  970--982.

\bibitem{AIM14}
S.~Aaronson, R.~Impagliazzo, D.~Moshkovitz, {AM} with multiple merlins, in:
  Proc.\ of CCC, 2014, pp. 44--55.

\bibitem{AIM14arx}
S.~Aaronson, R.~Impagliazzo, D.~Moshkovitz, {AM} with multiple merlins, CoRR
  abs/1401.6848.

\bibitem{FNS07}
T.~Feder, H.~Nazerzadeh, A.~Saberi, Approximating {N}ash equilibria using
  small-support strategies, in: Proc.\ of EC, 2007, pp. 352--354.

\bibitem{DFS16}
A.~Deligkas, J.~Fearnley, R.~Savani, Inapproximability results for approximate
  {N}ash equilibria, in: Proc.\ of WINE, 2016, pp. 29--43.

\bibitem{BPR16}
Y.~Babichenko, C.~H. Papadimitriou, A.~Rubinstein, Can almost everybody be
  almost happy?, in: Proc.\ of ITCS, 2016, pp. 1--9.

\bibitem{MR10}
D.~Moshkovitz, R.~Raz, Two-query {PCP} with subconstant error, J. {ACM} 57~(5)
  (2010) 29:1--29:29.

\bibitem{Dinur07}
I.~Dinur, The {PCP} theorem by gap amplification, J. {ACM} 54~(3) (2007)
  Article No.\ 12.

\bibitem{CFJ15}
A.~Czumaj, M.~Fasoulakis, M.~Jurdzi\'nski, Approximate {N}ash equilibria with
  near optimal social welfare, in: Proc.\ of IJCAI, 2015, pp. 504--510.

\bibitem{CFJ16}
A.~Czumaj, M.~Fasoulakis, M.~Jurdzi\'nski, Approximate plutocratic and
  egalitarian {N}ash equilibria: (extended abstract), in: Proc.\ of AAMAS,
  2016, pp. 1409--1410.

\bibitem{N51}
J.~Nash, Non-cooperative games, The Annals of Mathematics 54~(2) (1951)
  286--295.

\end{thebibliography}

\end{document}
